\documentclass[a4paper,UKenglish,cleveref,autoref,thm-restate]{lipics-v2021}

\hideLIPIcs  

\newcommand{\vone}{\mathrm{Out}^1}
\newcommand{\vtwo}{\mathrm{Out}^2}
\newcommand{\vl}{\mathrm{Out}^1_\mathrm{L}}
\newcommand{\vr}{\mathrm{Out}^1_\mathrm{R}}
\newcommand{\vlr}{\mathrm{Out}^1_\mathrm{LR}}
\newcommand{\vm}{\mathrm{Out}^1_\mathrm{M}}

\graphicspath{{./img/}}

\title{On Planar Straight-Line Dominance Drawings\footnote{A preliminary version of this paper appeared at WADS '25 \cite{DBLP:conf/wads/AngeliniBBF0O25}. Sections \ref{sse:deep} and \ref{se:span-2}, and thus Theorems \ref{th:deep-3-trees} and \ref{th:span2}, are new to this version.}}

\author{Patrizio Angelini}{John Cabot University, Rome, Italy}{pangelini@johncabot.edu}{}{}
\author{Michael A. Bekos}{University of Ioannina, Ioannina, Greece}{bekos@uoi.gr}{}{}
\author{Giuseppe {Di Battista}}{Universit\`a degli Studi Roma Tre, Rome, Italy}{giuseppe.dibattista@uniroma3.it}{}{Supported by the European Union, Next Generation EU, Mission~4, Component 1, CUP C53D23003680006 PRIN project no.\ 2022TS4Y3N ``EXPAND: scalable algorithms for EXPloratory Analyses of heterogeneous and dynamic Networked Data''.}
\author{Fabrizio Frati}{Universit\`a degli Studi Roma Tre, Rome, Italy}{fabrizio.frati@uniroma3.it}{}{Supported by the European Union, Next Generation EU, Mission~4, Component 1, CUP J53D23007130006 PRIN proj.\ 2022ME9Z78 ``NextGRAAL: Next-generation algorithms for constrained GRAph visuALization''.}
\author{Luca Grilli}{Universit\`a degli Studi di Perugia, Perugia, Italy}{luca.grilli@unipg.it}{}{Supported by MUR PNRR project SERICS (COVERT: CUP\_J93C23002310006), funded by the European Union – Next Generation EU}
\author{Giacomo Ortali}{Universit\`a degli Studi di Perugia, Perugia, Italy}{giacomo.ortali@unipg.it}{}{}

\authorrunning{P. Angelini, M. Bekos, G. Di Battista, F. Frati, L. Grilli, G. Ortali} 

\Copyright{Patrizio Angelini, Michael A. Bekos, Giuseppe Di Battista, Fabrizio Frati, Luca Grilli, Giacomo Ortali}

\keywords{$st$-graphs, dominance drawings, planar straight-line drawings, upward planarity}

\ccsdesc[500]{Theory of computation~Computational geometry}
\ccsdesc[500]{Mathematics of computing~Graph algorithms}

\nolinenumbers 

\usepackage{xspace}
\usepackage{todonotes}
\usepackage{thmtools} 
\usepackage{apptools}
\usepackage{complexity}
\usepackage{xcolor}

\crefname{property}{Property}{Properties}
\EventEditors{Pat Morin and Eunjin Oh}
\EventNoEds{2}
\EventLongTitle{19th International Symposium on Algorithms and Data Structures (WADS 2025)}
\EventShortTitle{WADS 2025}
\EventAcronym{WADS}
\EventYear{2025}
\EventDate{August 11--15, 2025}
\EventLocation{York University, Toronto, Canada}
\EventLogo{}
\SeriesVolume{349}
\ArticleNo{23}

\begin{document}

\maketitle
\begin{abstract}
We study the following question, which has been considered since the 90's: Does every $st$-planar graph admit a planar straight-line dominance drawing? We show concrete evidence for the difficulty of this question, by proving that, unlike upward planar straight-line drawings, planar straight-line dominance drawings with prescribed $y$-coordinates do not always exist and planar straight-line dominance drawings cannot always be constructed via a contract-draw-expand inductive approach. We also show several classes of $st$-planar graphs that always admit a planar straight-line dominance drawing. These include $st$-planar $3$-trees in which every stacking operation introduces two edges incoming into the new vertex, $st$-planar graphs in which every vertex is adjacent to the sink, $st$-planar graphs in which no face has the left boundary that is a single edge, and $st$-planar graphs that have a leveling with span at most two. 
\end{abstract}

\section{Introduction}

Drawings of directed graphs are an evergreen research topic in the graph drawing literature. Early papers on the subject go back to the 80's~\cite{DBLP:journals/tsmc/BattistaN88,DBLP:journals/tcs/BattistaT88,DBLP:journals/tsmc/SugiyamaTT81} and the number of papers on the topic published since 2023 is in double digits~\cite{DBLP:conf/gd/AlegriaCLDBFGP24,acd-upse-25,angelini_et_al:LIPIcs.GD.2024.13,DBLP:journals/dmtcs/AngeliniBCCLHLPPR25,DBLP:journals/jgaa/AngeliniCCL24,DBLP:conf/gd/BekosLFGMR22,DBLP:journals/tcs/BekosLFGMR23,DBLP:journals/ejc/BhoreLMN23,DBLP:journals/algorithmica/BinucciLGDMP23,DBLP:journals/jgaa/BinucciDP23,DBLP:journals/cgt/ChaplickCCLNPT023,DBLP:conf/latin/GiacomoFKMMSV24,DBLP:journals/tcs/Frati24,DBLP:conf/gd/JansenKKLMS23,DBLP:conf/focs/JungeblutMU23,DBLP:journals/siamdm/JungeblutMU23,DBLP:journals/jgaa/KlawitterZ23,loffler:LIPIcs.GD.2024.47,DBLP:conf/gd/NollenburgP23}. From an applicative perspective, many domains require techniques for visualizing directed graphs, such as visualization tools for biological networks and SIEM systems for cyber threat intelligence. Many standards for drawing directed graphs have been defined, and in most of them the drawing is \emph{upward}, i.e., each edge is represented by a Jordan arc whose $y$-coordinates monotonically increase from the tail to the head of the edge. Di Battista and Tamassia~\cite{DBLP:journals/tcs/BattistaT88} proved that every \emph{upward planar graph} (that is, a directed graph that admits an upward planar drawing) admits an upward planar \emph{straight-line} drawing, a result analogous to Fary's celebrated result about the geometric realizability of planar graphs~\cite{f-srpg-48}. In order to prove the geometric realizability of upward planar graphs, it suffices to look at upward planar graphs whose faces are delimited by $3$-cycles. Indeed, every upward planar graph is a subgraph of an \emph{$st$-planar graph}~\cite{DBLP:journals/tcs/BattistaT88} (that is, an upward planar graph with a single source $s$ and a single sink $t$), which in turn is a subgraph of a \emph{maximal $st$-planar graph}~\cite{DBLP:journals/tcs/BattistaT88} (that is, an $st$-planar graph to which no edge can be added without losing simplicity or upward planarity).

One of the easiest algorithms, if not \emph{the} easiest algorithm, for constructing upward planar straight-line drawings is due to Di Battista, Tamassia, and Tollis~\cite{DBLP:journals/dcg/BattistaTT92}. This algorithm assigns $x$- and $y$-coordinates to the vertices simply by performing two pre-order traversals of the input $st$-planar graph. Moreover, the algorithm constructs upward planar straight-line drawings that are actually \emph{dominance drawings}. These are \emph{$xy$-monotone drawings} (that is, each edge is represented by a Jordan arc whose $x$- and $y$-coordinates monotonically increase from the tail to the head of the edge) such that, for any pair of vertices $u,v$, there exists a directed path from $u$ to $v$ in the graph if and only if $x(u)\leq x(v)$ and $y(u)\leq y(v)$ hold in the drawing. Dominance drawings constitute an interesting graph drawing style because they express the reachability between vertices by their dominance relationship, i.e., by the coordinates assigned to them; this allows one to answer reachability queries in constant time, see, e.g,~\cite{DBLP:journals/sncs/LionakisOT21,DBLP:conf/edbt/VelosoCJZ14}. For more about dominance drawings, see, e.g.,~\cite{DBLP:conf/analco/BannisterDE14,DBLP:journals/ijcga/BertolazziCBTT94,DBLP:conf/cccg/ElGindyHLMRW93,DBLP:conf/gd/KornaropoulosT12a,ot-mddta-21,DBLP:conf/sofsem/OrtaliT23,DBLP:journals/tsmc/SugiyamaTT81}. Figure~\ref{fig:introduction-drawings} shows planar straight-line drawings of an $st$-planar graph that are non-upward, upward (and not $xy$-monotone), $xy$-monotone (and not dominance), and dominance.

\begin{figure}[ht]
    \centering
    \begin{subfigure}{0.24\linewidth}
        \centering
        {\includegraphics[page=1]{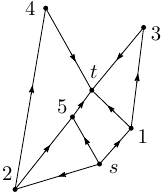}}
        \subcaption{Non-upward}\label{fig:introduction-drawings-a}
    \end{subfigure}
    \begin{subfigure}{0.24\linewidth}
        \centering
        {\includegraphics[page=2]{img/introduction-drawings2.pdf}}
        \subcaption{Upward}\label{fig:introduction-drawings-b}
    \end{subfigure}
    \begin{subfigure}{0.24\linewidth}
        \centering
        {\includegraphics[page=3]{img/introduction-drawings2.pdf}}
        \subcaption{$xy$-monotone}\label{fig:introduction-drawings-c}
    \end{subfigure}
    \begin{subfigure}{0.24\linewidth}
        \centering
        {\includegraphics[page=4]{img/introduction-drawings2.pdf}}
        \subcaption{Dominance}\label{fig:introduction-drawings-d}
    \end{subfigure}
    \caption{Four planar straight-line drawings of an $st$-planar graph $G$. (a) A non-upward drawing. (b) An upward drawing. (c) An $xy$-monotone drawing. (d) A dominance drawing.}
    \label{fig:introduction-drawings}
\end{figure}

Di Battista, Tamassia, and Tollis's algorithm~\cite{DBLP:journals/dcg/BattistaTT92}  does not actually construct an upward planar straight-line drawing of every $st$-planar graph. Indeed, it may construct a non-planar drawing if the input $st$-planar graph contains \emph{transitive} edges, where an edge is transitive if the graph contains a directed path from the tail to the head of the edge. By subdividing each transitive edge with a new vertex, their algorithm constructs a planar dominance drawing of any $st$-planar graph in which each edge is either a straight-line segment (if it is non-transitive) or a $1$-bend polyline (if it is transitive). Whether this bend per edge can be eliminated by designing an algorithm different from the one in~\cite{DBLP:journals/dcg/BattistaTT92} is the question we study in this paper. 

Formally, we ask: Does every $st$-planar graph admit a planar straight-line dominance drawing? Apart from the $st$-planar graphs without transitive edges, the question is known to have a positive answer for \emph{series-parallel digraphs}~\cite{DBLP:journals/ijcga/BertolazziCBTT94}. We prove the following results.

\begin{itemize}
\item In Section~\ref{se:barriers}, we prove a remarkable difference between dominance and upward drawings. We revisit the two main approaches for the construction of upward planar straight-line drawings of $st$-planar graphs and prove that they cannot be successfully applied to construct planar straight-line dominance drawings. The first approach~\cite{DBLP:journals/tcs/BattistaT88} contracts an internal edge of the graph, constructs a drawing inductively, and then expands the previously contracted edge to be a ``short'' segment. We show that there exist $st$-planar graphs in which no edge can be used in the contract-draw-expand approach so to get a planar straight-line dominance drawing. The second method~\cite{DBLP:journals/algorithmica/EadesFLN06,DBLP:journals/jda/HongN10,DBLP:journals/jgt/PachT04} prescribes the $y$-coordinates of the vertices, so that the tail of any edge is assigned a smaller $y$-coordinate than its head. This additional constraint on the drawing allows for easier recursive schemes for its construction. We prove that planar straight-line dominance drawings with prescribed $y$-coordinates do not always exist. We believe that these results provide solid evidence for the difficulty of constructing planar straight-line dominance drawings.
\item In Section~\ref{se:$3$-trees}, we study $st$-planar graphs whose underlying graph is a planar $3$-tree. Planar $3$-trees, also known as stacked triangulations and Apollonian networks, constitute a common benchmark for planar graph drawing problems, as they allow for easy inductive constructions; for example, every planar $3$-tree with at least four vertices can be constructed by ``stacking'' a vertex inside a face of a smaller planar $3$-tree. For our question, the study of $st$-planar $3$-trees turns out to be complicated, as inductive drawing constructions do not cope well with the dominance relationship that needs to be ensured between  vertices that are ``far away'' in the graph. We show how to construct planar straight-line dominance drawings for three classes of $st$-planar $3$-trees. One of them has a constraint on the orientation, namely every stacking operation introduces two edges incoming into the new vertex. The other two classes have constraints on the graph structure, namely in one of them each stacking operation happens in a face that was created by the last stacking operation, while in the other one every stacking operation happens in a face incident to the sink. The latter graph class coincides with the maximal  $st$-planar graphs in which the sink is adjacent to every vertex.
\item In Section~\ref{se:non-transitive}, we improve the mentioned result by Di Battista, Tamassia, and Tollis~\cite{DBLP:journals/dcg/BattistaTT92}, by proving that a planar straight-line dominance drawing always exists for any $st$-planar graph in which every transitive edge is to the right of every directed path from the tail to the head of the edge. This result is obtained via an ear decomposition of the graph. This shows that the problem of constructing planar straight-line dominance drawings is made difficult by the interaction between ``left transitive edges'' and ``right transitive edges''. 
\item In Section~\ref{se:span-2}, we show that every $st$-planar graph with only ``short'' edges admits a planar straight-line dominance drawing. A natural measure of length for the edges of an $st$-planar graph is called \emph{span} and is defined as follows. An \emph{$st$-planar level graph} consists of an $st$-planar graph $G=(V,E)$ together with a function $\ell: V\rightarrow \{1,2,\dots,k\}$, for some integer $k$, such that $G$ admits an upward planar drawing in which $y(u)=\ell(u)$, for every vertex $u\in V$. Observe that every $st$-planar graph can be enhanced to an $st$-planar level graph by defining a suitable function $\ell$. The \emph{span} of an $st$-planar level graph $G$ is the maximum value $\ell(v)-\ell(u)$, among all the edges $(u,v)$ of $G$. Note that $st$-planar level graphs with span~one do not have transitive edges, hence they admit planar straight-line dominance drawings~\cite{DBLP:journals/dcg/BattistaTT92}. We show that $st$-planar level graphs with span~two also admit planar straight-line dominance drawings.

\end{itemize}


All our algorithms construct drawings whose resolution is exponentially small (or worse). This drawback is sometimes necessary for upward planar straight-line drawings \cite{DBLP:journals/dcg/BattistaTT92}, and hence also for planar straight-line dominance drawings. However, for most of the graph classes we considered, we do not know whether an exponentially-small resolution is actually required in order to construct planar straight-line dominance drawings.


\section{Preliminaries} \label{se:preliminaries}

A \emph{drawing} of a graph maps each vertex to a distinct point in the plane and each edge to a Jordan arc between its endpoints. A drawing is \emph{straight-line} if each edge is represented by a straight-line segment and \emph{planar} if no two edges intersect, except at common endpoints. Two planar drawings of a connected graph are \emph{plane-equivalent} if they define the same clockwise order of the edges incident to each vertex and the same clockwise order of the vertices and edges along the boundary of the outer face. A \emph{plane embedding} is an equivalence class of planar drawings and a \emph{plane graph} is a graph with a plane embedding. Whenever we talk about planar drawings of a plane graph, we always assume that they are in the equivalence class associated with the plane graph. A \emph{face} is a connected region of the plane defined by a planar drawing. Bounded faces are \emph{internal}, while the unbounded face is the \emph{outer face}. Vertices incident to the outer face are \emph{external}, while the other vertices are \emph{internal}. We often talk about faces of a plane embedding or of a plane graph, implicitly referring to any planar drawing in the corresponding equivalence class.

An \emph{$st$-plane graph} is an $st$-planar graph with a plane embedding (for its underlying graph) in which $s$ and $t$ are incident to the outer face. An $st$-plane graph is \emph{maximal} if no edge can be added to it while maintaining it an $st$-plane graph. Since every $st$-planar graph can be augmented (by adding vertices and edges) to maximal without altering the reachability between vertices~\cite{DBLP:journals/tcs/BattistaT88}, the existence of a planar straight-line dominance drawing for all $st$-planar graphs can be decided by only looking at maximal $st$-planar graphs. Two vertices in a directed graph are \emph{incomparable} if no directed path goes from any of the vertices to the other one, that is, neither is \emph{reachable} from the other.

All our algorithms, except for the one presented in the proof of Theorem~\ref{th:span2}, construct planar straight-line dominance drawings in which no two vertices share the same $x$- or $y$-coordinates. The next lemma shows that, in fact, equality between coordinates is not needed for constructing planar straight-line dominance drawings.

\begin{lemma} \label{le:strict}
If a directed graph admits a planar straight-line dominance drawing, it also admits a planar straight-line dominance drawing in which no two vertices share the same $x$- or $y$-coordinate.
\end{lemma}

\begin{proof}
Let $\Gamma$ be a planar straight-line dominance drawing of a directed graph $G$. We construct the desired drawing by induction on the number $p$ of pairs of vertices that share their $x$- or $y$-coordinates in $\Gamma$. In the base case, $p=0$ and $\Gamma$ is the desired drawing. Suppose now that $p>0$. Since $\Gamma$ is a dominance drawing, for any two vertices with the same $x$-coordinate there exists a directed path from one to the other all of whose vertices have the same $x$-coordinate, and similar for two vertices with the same $y$-coordinate. The last edge $(u,v)$ of a maximal directed path whose vertices have the same $x$- or $y$-coordinates is such that either: 
\begin{itemize}
\item $x(u)=x(v)$, $y(u)<y(v)$, and for any successor $w$ of $v$ we have $x(w)> x(v)$, or
\item $y(u)=y(v)$, $x(u)<x(v)$, and for any successor $w$ of $v$ we have $y(w)> y(v)$.
\end{itemize}
Suppose we are in the former case, as the discussion for the latter case is symmetric. We can then increment the $x$-coordinate of $v$ by a sufficiently small amount $\varepsilon>0$ so that: (i) $\Gamma$ remains planar; the fact that each vertex position can be perturbed (actually in any direction) while maintaining planarity is a standard argument, see e.g.~\cite{DBLP:journals/tcs/BattistaT88}; (ii) $\Gamma$ remains a dominance drawing; for this, it suffices to choose $\varepsilon$ sufficiently small so that no vertex $w\neq v$ has $x$-coordinate in the interval $(x(u),x(u)+\varepsilon]$. In particular, since by assumption no edge outgoing from $v$ is vertical, all such edges maintain a positive slope. In the resulting drawing~$\Gamma$, the number of pairs of vertices that share their $x$ or $y$-coordinates is less than $p$, hence this completes the induction.
\end{proof}

As a warm-up result, we prove that every Hamiltonian $st$-planar graph has a planar straight-line dominance drawing. A directed graph is \emph{Hamiltonian} if it contains a directed path $(v_1=s,v_2,\dots,v_n=t)$, where $\{v_1,v_2,\dots,v_n\}$ is the vertex set of the graph. 

\begin{theorem} \label{th:hamiltonian}
Hamiltonian $st$-planar graphs admit planar straight-line dominance drawings.
\end{theorem}

\begin{proof}
Consider a Hamiltonian $st$-planar graph $G$. Construct an upward planar straight-line drawing $\Gamma$ of $G$; this always exists~\cite{DBLP:journals/tcs/BattistaT88}. Stretch $\Gamma$ vertically, so that the slope of every edge is in the range $(45^{\circ}, 135^{\circ})$. Now rotate $\Gamma$ in clockwise direction by $45^{\circ}$. Since the slope of every edge is now in the range $(0^{\circ}, 90^{\circ})$, we have that $\Gamma$ is $xy$-monotone. Since vertical stretch and rotation are affine transformations, $\Gamma$ is planar, as well. Finally, since $G$ contains a Hamiltonian path $(v_1,\dots,v_n)$, vertex $v_j$ is a successor of vertex $v_i$, for every $1\leq i<j\leq n$. Since the slope of the edge $(v_k,v_{k+1})$ is in the range $(0^{\circ}, 90^{\circ})$, for $k=i,\dots,j-1$, vertex $v_j$ is in the first quadrant of vertex $v_i$, hence $\Gamma$ is a dominance drawing.  
\end{proof}


\section{Planar Straight-line Dominance Drawings are Difficult to Get} \label{se:barriers}

In this section,  we revisit the two main approaches for the construction of upward planar straight-line drawings of $st$-planar graphs and prove that they cannot be enhanced (or at least not in a direct way) to construct planar straight-line dominance drawings.


\subsection{Constructing Drawings via Contractions and Expansions}\label{subse:fary}

Di Battista and Tamassia~\cite{DBLP:journals/tcs/BattistaT88} first proved that every $st$-plane graph admits an upward planar straight-line drawing. 
Their proof extends to directed graphs a well-known proof by F\'ary~\cite{f-srpg-48}, showing that every (undirected) plane graph admits a planar straight-line drawing. We briefly describe the algorithm by Di Battista and Tamassia~\cite{DBLP:journals/tcs/BattistaT88}.

An internal edge $(u,v)$ of a maximal $st$-plane graph $G$ is \emph{contractible} if it satisfies the following conditions: (1) The vertices $u$ and $v$ have exactly two common neighbors, denoted by $z_1$ and $z_2$; note that the cycles $(u,v,z_1)$ and $(u,v,z_2)$ delimit internal faces of $G$. (2) For $i=1,2$, the edges connecting $u$ and $v$ with $z_i$ are both incoming or both outgoing at $z_i$. 

The \emph{contraction} of a contractible edge $(u,v)$ constructs a graph $G'$ by identifying $u$ and $v$ into a vertex $w$ with the following adjacencies (see Fig \ref{fig:fary-contraction}). For every neighbor $z\notin \{u,v,z_1,z_2\}$ of $u$ (of $v$), we have that $G'$ contains an edge between $w$ and $z$, which is outgoing at~$z$ if and only if the edge between $u$ and $z$ (resp.\  between $v$ and $z$) is outgoing at~$z$. Also, for $i=1,2$, we have that $G'$ contains an edge between $w$ and~$z_i$, which is outgoing at~$z_i$ if and only if the edges connecting $u$ and $v$ with $z_i$ are both outgoing at $z_i$. It is easy to see that $G'$ is a maximal $st$-plane graph.

The core of Di Battista and Tamassia's algorithm lies in the following two statements\footnote{Di Battista and Tamassia's proof actually distinguishes the case in which $G$ contains a \emph{separating triangle} (a $3$-cycle with vertices in its interior) from the case in which it does not, performing different constructions in the two cases. However, the first case is unnecessary, as a contractible edge in an $st$-planar graph can always be found, similarly to what was noted by Wood~\cite{DBLP:journals/corr/abs-cs-0505047} for undirected graphs.}: (i) every maximal $st$-plane graph $G$ has a contractible edge $(u,v)$, whose contraction results in a maximal $st$-plane graph $G'$; and (ii) an upward planar straight-line drawing $\Gamma$ of $G$ can be obtained from an upward planar straight-line drawing $\Gamma'$ of $G'$ by \emph{expanding} $w$, that is, by replacing $w$ with a sufficiently small segment (with a suitable slope) representing $(u,v)$.  

\begin{figure}[ht]
    \centering
    \begin{subfigure}{0.5\linewidth}
        \centering
        {\includegraphics[page=1]{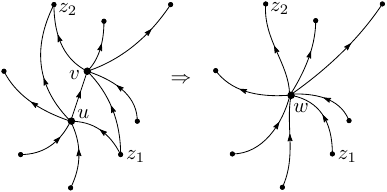}
        \subcaption{}\label{fig:fary-contraction}}
    \end{subfigure}
    \hfill
    \begin{subfigure}{0.4\linewidth}
        \centering
        {\includegraphics[scale=0.9,page=2]{img/Fary2.pdf}
        \subcaption{}\label{fig:fary-expansion}}
    \end{subfigure}    
    \caption{(a) The contraction of an edge $(u,v)$ in a maximal $st$-plane graph. (b) A maximal $st$-plane graph with no dominance-expandable edge. Thin edges are not contractible, while fat edges are contractible but not dominance-expandable; for example, $(1,3)$ is not dominance-expandable, because vertex $2$ is a predecessor of vertex $3$ but not a predecessor of vertex $1$.}
    \label{fig:fary}
\end{figure}

Since, depending on the geometric placement of the neighbors of $w$ in $\Gamma'$, the edge $(u,v)$ might need to be an arbitrarily small segment in $\Gamma$, in order for $\Gamma$ to be a dominance drawing we need $u$ and $v$ to have the same successors and predecessors. That is, let $\mathcal S(z)$ be the set of successors of a vertex $z$, that is, the set of all vertices $z'$ such that there exists a directed path from $z$ to $z'$. Analogously, let $\mathcal P(z)$ be the sets of predecessors of a vertex $z$. A contractible edge $(u,v)$ is \emph{dominance-expandable} if $\mathcal S(u)=\mathcal S(v) \cup \{v\}$ and $\mathcal P(v)=\mathcal P(u) \cup \{u\}$. Di Battista and Tamassia's approach could be enhanced to construct planar straight-line dominance drawings if every maximal $st$-plane graph contained a dominance-expandable edge. However, we can prove that there exist maximal $st$-plane graphs with no dominance-expandable edge, as the one in Fig~\ref{fig:fary-expansion}, which constitutes a barrier for this approach we cannot overcome. 

We remark that, for every graph class for which we can prove the existence of planar straight-line dominance drawings in the upcoming sections, there exist graphs in the class that do not have a dominance-expandable edge or such that the contraction of any dominance-expandable edge would result in a graph not in the same class.


\subsection{Constructing Drawings by Prescribing the $y$-Coordinates}\label{subse:prescribed}

Eades, Feg, Lin, and Nagamochi~\cite{DBLP:journals/algorithmica/EadesFLN06} and, independently,  Pach and T{\'{o}}th~\cite{DBLP:journals/jgt/PachT04} proved that every upward planar drawing can be straightened while preserving the $y$-coordinates of the vertices. This implies that every $st$-plane graph admits an upward planar straight-line drawing with prescribed $y$-coordinates (as long as these respect the reachability between vertices). This result was strengthened by Hong and Nagamochi~\cite{DBLP:journals/jda/HongN10}, who proved that every internally-triconnected $st$-plane graph admits an upward planar straight-line \emph{convex} drawing with prescribed $y$-coordinates and prescribed outer face. It is interesting that, while more constrained, drawings with prescribed $y$-coordinates (and a prescribed outer face) allow for an easier recursive construction.

We now show that, unlike upward planar straight-line drawings, planar straight-line dominance drawings with given $y$-coordinates do not always exist. 


\begin{theorem} \label{th:fixed-y-counter}
For every $n\geq 7$, there exists an $st$-planar graph $G_n$ with vertex set $\{v_1,v_2,\dots,v_n\}$ such that:
\begin{itemize}
	\item there exists a planar dominance drawing of $G_n$ such that $y(v_i)=i$, for $i=1,\dots,n$; and 
	\item there exists no planar straight-line dominance drawing of $G_n$ such that $y(v_i)=i$, for $i=1,\dots,n$.
\end{itemize}
\end{theorem}

\begin{figure}[ht]
    \centering
    \hfill
    \begin{subfigure}{0.4\linewidth}
        \centering
        {\includegraphics[width=.9\textwidth,page=1]{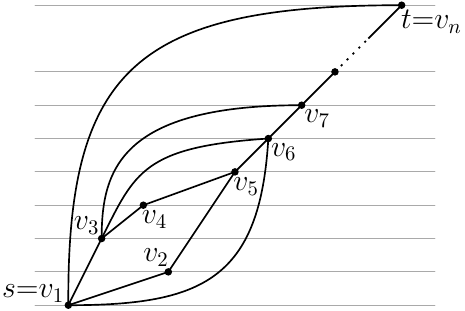}
        \subcaption{\label{fig:fixed-y-construction}}}
    \end{subfigure}
    \hfill
    \begin{subfigure}{0.4\linewidth}
        \centering
        {\includegraphics[width=.9\textwidth,page=2]{Fixed-y.pdf}
        \subcaption{\label{fig:fixed-y-proof}}}
    \end{subfigure}        
    \hfill
    \caption{(a) The graph for the proof of Theorem~\ref{th:fixed-y-counter}. (b) The rays $\ell_{1,2}$ and $\ell_{3,4}$ diverge.}
    \label{fig:fixed-y}
\end{figure}

\begin{proof}
The $st$-planar graph $G_n$ consists of the directed paths $(s=v_1,v_2,v_5)$, $(v_1,v_3,v_4,v_5)$, $(v_5,v_6,\dots,v_n=t)$, and of the edges $(v_1,v_6)$, $(v_3,v_6)$, $(v_3,v_7)$, and $(v_1,v_n)$. Fig~\ref{fig:fixed-y-construction} shows a planar dominance drawing of $G_n$ with $y(v_i)=i$, for $i=1,\dots,n$. Suppose, for a contradiction, that a planar straight-line dominance drawing $\Gamma$ of $G_n$ exists with $y(v_i)=i$, for $i=1,\dots,n$. 
We prove that the plane embedding in $\Gamma$ of the underlying graph of $G_n$ is the one in Fig~\ref{fig:fixed-y-construction}, except, possibly, for the position of the edge $(s,t)$. Obviously, the path $(v_1,v_2,v_5,v_6,\dots,v_n)$ has a unique plane embedding. Since $v_2$ and $v_4$ are incomparable and $y(v_2)<y(v_4)$, we have $x(v_4)\leq x(v_2)$, hence the clockwise order of the vertices along the cycle $\mathcal C:=(v_1,v_3,v_4,v_5,v_2)$ is $v_1,v_3,v_4,v_5,v_2$. From that, we get that the edges $(v_3,v_6)$ and $(v_3,v_7)$ lie above the path $(v_3,v_4,v_5,v_6,v_7)$, and finally that the edge $(v_1,v_6)$ lies below the path $(v_1,v_2,v_5,v_6)$. 
%
%
For any distinct $i,j \in \{1,\dots,n\}$, let $\ell_{i,j}$ be the ray starting at $v_i$ and passing through $v_j$. Since $x(v_1)<x(v_3)<x(v_4)\leq x(v_2)$, we have $x(v_2)-x(v_1)>x(v_4)-x(v_3)$. Also, we have $y(v_2)-y(v_1)=y(v_4)-y(v_3)=1$. Hence, the ray $\ell_{1,2}$ has smaller slope than the ray $\ell_{3,4}$; that is, such rays diverge, see Fig~\ref{fig:fixed-y-proof}. Since the ray $\ell_{1,6}$ has smaller slope than $\ell_{1,2}$, and since the ray $\ell_{3,6}$ has larger slope than $\ell_{3,4}$, it follows that $\ell_{1,6}$ and $\ell_{3,6}$ also diverge, while they meet at $v_6$, a contradiction which proves the theorem. Note that vertices $v_8,\dots,v_n$ only serve the purpose of creating an infinite graph family. 
\end{proof}

We can similarly show that one cannot, in general, prescribe the $x$-coordinates of a planar straight-line dominance drawing.  

Also, we can strengthen Theorem~\ref{th:fixed-y-counter} by proving that, for every $n\geq 10$ and for every sequence $y_1<\dots<y_n$ of $y$-coordinates, there exists an $st$-planar graph $G'_n$ with vertex set $\{v_1,\dots,v_n\}$ such that there exists a planar dominance drawing of $G'_n$ with $y(v_i)=y_i$, for $i=1,\dots,n$, and there exists no planar straight-line dominance drawing of $G'_n$ with $y(v_i)=y_i$, for $i=1,\dots,n$. That is, the $y$-coordinates prescribed by Theorem~\ref{th:fixed-y-counter} do not need to be uniformly distributed. 

The key point for this is the observation that the proof of Theorem~\ref{th:fixed-y-counter} works as long as $y(v_2)-y(v_1)\leq y(v_4)-y(v_3)$. Hence, we can consider the four lines $y=y_i$, with $i=4,5,6,7$, and then distinguish two cases. If $y_5-y_4\leq y_7-y_6$, we let our $st$-planar graph $G'_n$ contain the graph $G_7$ from the proof of Theorem~\ref{th:fixed-y-counter} and we set $y(v_i)=y_{i+3}$, for $i=1,\dots,7$, where $v_1,\dots,v_7$ is the vertex set of $G_7$. Otherwise, that is, if $y_7-y_6< y_5-y_4$, we let our $st$-planar graph $G'_n$ contain the graph obtained by reversing the edge directions of the graph $G_7$ and we set $y(v_i)=y_{8-i}$, for $i=1,\dots,7$, where $v_1,\dots,v_7$ is the vertex set of $G_7$. 


\section{$st$-plane $3$-trees} \label{se:$3$-trees}

A \emph{plane $3$-tree} is a plane graph recursively defined as follows. A $3$-cycle embedded in the plane is a plane $3$-tree. Any plane $3$-tree with $n\geq 4$ vertices can be obtained from a plane $3$-tree with $n-1$ vertices by \emph{stacking} a new vertex into an internal face, that is, by connecting the new vertex to the three vertices incident to the face. An \emph{$st$-plane $3$-tree} is an $st$-plane graph whose underlying graph is a plane $3$-tree. In our opinion, $st$-plane $3$-trees constitute a very challenging class of $st$-plane graphs for our problem. Indeed, the ``natural'' strategies for drawing the graphs in this class are to either recursively construct and then combine the drawings of three smaller $st$-plane $3$-trees, or to iteratively add a single vertex to a previously constructed drawing of a smaller $st$-plane $3$-tree; both these strategies do not cope well with the geometric relationship that has to be ensured for incomparable vertices. Nevertheless, in this section we show how to obtain planar straight-line dominance drawings of three classes of $st$-plane $3$-trees.


\subsection{Upper $st$-plane $3$-trees}

Consider the construction of an $st$-plane $3$-tree $G$ via repeated stacking operations. If a vertex $u$ is stacked into a face delimited by a cycle $(a,b,c)$, where $a$ and $c$ are the source and the sink of the cycle, respectively, then the edge $(a,u)$ is directed from $a$ to $u$, the edge $(u,c)$ is directed from $u$ to $c$, while the edge $(b,u)$ might be directed either way. We say that $G$ is an \emph{upper $st$-plane $3$-tree} if, at every stacking operation, the edge that can be directed either way is always directed towards the newly inserted vertex. We have the following. 

\begin{theorem} \label{th:upper-3-trees}
Upper $st$-plane $3$-trees admit planar straight-line dominance drawings.
\end{theorem}

\begin{proof}
Let $G$ be an $n$-vertex upper $st$-plane $3$-tree whose outer face is delimited by the cycle $(s,m,t)$. Let $\Delta$ be any triangle with vertices $p_s,p_m,p_t$, where $x(p_s)<x(p_m)<x(p_t)$ and $y(p_s)<y(p_m)<y(p_t)$. Also, let $D$ be a closed disk in the interior of $\Delta$ such that, for any point $p$ in $D$, we have  $x(p_m)<x(p)<x(p_t)$ and $y(p_m)<y(p)<y(p_t)$; see Figs~\ref{fig:upper-input-1} and~\ref{fig:upper-input-2}. We prove by induction that $G$ admits a planar straight-line dominance drawing such that:
\begin{itemize}
    \item $s$ lies at $p_s$, $m$ lies at $p_m$, and $t$ lies at $p_t$; and
    \item every internal vertex of $G$ lies in the interior of $D$.
\end{itemize}

\begin{figure}[t]
    \centering
    \begin{subfigure}{0.24\linewidth}
        \centering
        {\includegraphics[width=.9\textwidth,page=1]{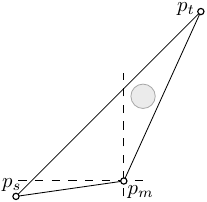}
        \subcaption{\label{fig:upper-input-1}}}
    \end{subfigure}
    \hfill
    \begin{subfigure}{0.24\linewidth}
        \centering
        {\includegraphics[width=.9\textwidth,page=3]{Upper.pdf}
        \subcaption{\label{fig:upper-subd-1}}}
    \end{subfigure}
    \hfill
    \begin{subfigure}{0.24\linewidth}
        \centering
        {\includegraphics[width=.9\textwidth,page=2]{Upper.pdf}
        \subcaption{\label{fig:upper-input-2}}}
    \end{subfigure}
    \hfill
    \begin{subfigure}{0.24\linewidth}
        \centering
        {\includegraphics[width=.95\textwidth,page=4]{Upper.pdf}
        \subcaption{\label{fig:upper-subd-2}}}
    \end{subfigure}
    \caption{(a) and (c) Triangle $\Delta$ and disk $D$ for the input to the induction. (b) and (d) Placing point $p_r$ (white) and disks $D_1$, $D_2$, and $D_3$ (gray) inside $D$; an enlarged view of the placement of $r$ and of disks $D_1$, $D_2$, and $D_3$ inside $D$ is also shown.}
	\label{fig:upper}
\end{figure}

The statement clearly implies the theorem. In the base case, in which $n=3$, the triangle $\Delta$ is the required drawing of $G$ and the statement is trivially true. 

Suppose now that $n>3$. Let $r$ be the first stacked vertex in the construction of $G$; that is, $r$ is the unique vertex of $G$ adjacent to $s$, $m$, and $t$. Note that the edge $(m,r)$ is directed away from $m$, given that $G$ is an upper $st$-plane $3$-tree. Let $G_1$, $G_2$, and $G_3$ be the subgraphs of $G$ inside the cycles $(s,m,r)$, $(s,r,t)$, and $(m,r,t)$, respectively. Note that $G_1$ is an upper $sr$-plane $3$-tree, $G_2$ is an upper $st$-plane $3$-tree, and $G_3$ is an upper $mt$-plane $3$-tree. Also, each of $G_1$, $G_2$, and $G_3$ has less than $n$ vertices. Let $p_r$ be any point inside $D$ and let $\Delta_1$, $\Delta_2$, and $\Delta_3$ be the triangles $(p_s,p_m,p_r)$, $(p_s,p_r,p_t)$, and $(p_m,p_r,p_t)$, respectively. Place $r$ at $p_r$; by the properties of $D$, we have $x(p_m)<x(p_r)$ and $y(p_m)<y(p_r)$, which complies with the orientation of $(m,r)$. Let $D_1$, $D_2$, and $D_3$ be closed disks such that (see Figs~\ref{fig:upper-subd-1} and~\ref{fig:upper-subd-2}):

\begin{itemize}
\item disk $D_1$ lies in the interior of $\Delta_1\cap D$, disk $D_2$ lies in the interior of $\Delta_2\cap D$, and disk $D_3$ lies in the interior of $\Delta_3\cap D$; 
\item for any point $p \in D_2 \cup D_3$, we have $x(p_r)<x(p)$ and $y(p_r)<y(p)$; and
\item for any point $p_2\in D_2$ and any point $p_3\in D_3$, if the clockwise order of the vertices of $\Delta$ is $p_s,p_t,p_m$, then  we have $x(p_2)<x(p_3)$ and $y(p_3)<y(p_2)$, otherwise we have $y(p_2)<y(p_3)$ and $x(p_3)<x(p_2)$.
\end{itemize}

Clearly, disks $D_1$, $D_2$, and $D_3$ with the above properties always exist. By induction, $G_1$, $G_2$, and $G_3$ have  planar straight-line dominance drawings $\Gamma_1$, $\Gamma_2$, and $\Gamma_3$ with $s$, $m$, $r$, and $t$ drawn at $p_s$, $p_m$, $p_r$, and $p_t$, respectively, so that the internal vertices of $G_1$, $G_2$, and $G_3$ lie in the interior of $D_1$, $D_2$, and $D_3$, respectively. This results in a straight-line drawing $\Gamma$ of $G$. 

Since $p_r$, $D_1$, $D_2$, and $D_3$ lie in the interior of $D$, all the internal vertices of $G$ lie in the interior of $D$, as required. The upward planarity of $\Gamma$ follows from the ones of $\Gamma_1$, $\Gamma_2$, and $\Gamma_3$. In order to prove that $\Gamma$ is a dominance drawing, consider any pair of vertices $u$ and $v$. 
\begin{itemize}
\item If $u$ and $v$ belong to the same graph $G_i$, for some $i\in \{1,2,3\}$, then their placement complies with their dominance relationship, by induction.
\item If one of $u$ and $v$ is $s$, say $u=s$, then $u$ is a predecessor of $v$, and indeed we have $x(u)<x(v)$ and $y(u)<y(v)$. The case $u=t$ can be discussed similarly.
\item If neither of $u$ and $v$ is $s$ or $t$, and one of $u$ and $v$ is $m$, say $u=m$, then $u$ is a predecessor of $v$, since $G$ is an upper $st$-plane $3$-tree. Since $x(p_m)<x(p)$ and $y(p_m)<y(p)$, for any point $p \in D$, we have $x(u)<x(v)$ and $y(u)<y(v)$.  
\item If $u$ is an internal vertex of $G_1$ and $v$ is an internal vertex of $G_2$ or $G_3$, then $u$  is a predecessor of $v$, since $G$ is an upper $st$-plane $3$-tree. Since, for any point $p \in D_1$ and any point $q\in D_2\cup D_3$, we have $x(p)<x(p_r)<x(q)$ and $y(p)<x(y_r)<y(q)$, the placement of $u$ and $v$ complies with their dominance relationship. 
\item Finally, if $u$ is an internal vertex of $G_2$ and $v$ is an internal vertex of $G_3$, then $u$ and $v$ are incomparable, since $G$ is an upper $st$-plane $3$-tree. Since, for any point $p_2 \in D_2$ and any point $p_3\in D_3$, we have $x(p_2)<x(p_3)$ and $y(p_3)<y(p_2)$, or $y(p_2)<y(p_3)$ and $x(p_3)<x(p_2)$, the placement of $u$ and $v$ complies with their dominance relationship.
\end{itemize}
This completes the induction and the proof of the theorem.
\end{proof}

An analogous result holds true for $st$-plane $3$-trees such that, at every stacking operation, the edge that can be directed either way is always directed out of the newly inserted vertex. 

Trying to use a similar strategy in order to construct a planar straight-line dominance drawing of every $st$-plane $3$-tree might be tempting. However, the ``types'' of internal vertices in a general $st$-plane $3$-tree are more than three. Namely, referring to the notation introduced in the proof of the theorem, the internal vertices of $G_2$ and $G_3$ are not all successors of $r$, but rather some are predecessors, some are successors, and some are incomparable to $r$. Hence, the ``three-disks schema'' fails, and more complex geometric invariants seem to be needed.


\subsection{Deep $st$-plane $3$-trees} \label{sse:deep}

We next consider the class of $st$-plane $3$-trees whose construction via repeated stacking operations has the following property: If a stacking operation stacks a vertex $u$ into a face delimited by a cycle $(a,b,c)$, then the next stacking operation happens into one of the three faces delimited by the cycles $(a,b,u)$, $(a,u,c)$, and $(u,b,c)$. If one represents the structure of an $st$-plane $3$-tree by a ternary tree, whose internal nodes correspond to non-facial $3$-cycles, whose leaves correspond to facial $3$-cycles, and whose edges represent containment between the corresponding cycles, then the $n$-vertex $st$-plane $3$-trees in the class under consideration are those whose corresponding ternary tree has maximum height, namely~$n-2$; because of this, we call {\em deep} the $st$-plane $3$-trees in this class. We have the following. 

\begin{theorem} \label{th:deep-3-trees}
Deep $st$-plane $3$-trees admit planar straight-line dominance drawings.
\end{theorem}

\begin{proof}
Let $G$ be an $n$-vertex deep $st$-plane $3$-tree, for some $n\geq 4$. Let $v_1,v_2,\dots,v_n$ be the vertices of $G$, ordered as follows. Vertices $v_1$, $v_2$, and $v_3$ are the external vertices of $G$, ordered in any way. Vertices $v_4,v_5,\dots,v_n$ are the internal vertices of $G$, ordered as they are introduced in $G$ in its construction via repeated stacking operations. 


We construct a drawing of $G$ ``backwards''. That is, we start from a drawing of the complete graph induced by $v_n$ and its neighbors, and then we insert in the drawing the remaining vertices one by one. Each vertex is inserted in the outer face of the current drawing and, at any time, the drawn graph is an $st$-plane $3$-tree which is a subgraph of $G$. More precisely, we maintain a drawing $\Gamma$ of a subgraph $H$ of $G$ with the following properties.

\begin{itemize}
    \item {(D1)} $H$ is a deep $st$-plane $3$-tree;
    \item {(D2)} let $\mathcal C_H$ denote the cycle delimiting the outer face of $H$; then all the vertices of $G$ inside $\mathcal C_H$ belong to $H$; and
    \item {(D3)} $\Gamma$ is a  planar straight-line dominance drawing of $H$.
\end{itemize}

We remark here that the order in which the vertices are inserted into $\Gamma$ is not, in general, the reverse of the order $v_1,v_2,\dots,v_n$ defined above.

The drawing $\Gamma$ is initialized as a planar straight-line dominance drawing of the subgraph~$H$ of $G$ induced by $v_n$ and its neighbors. Since, in the construction of $G$ via repeated stacking operations, the vertex $v_n$ is stacked into a face delimited by a $3$-cycle, the graph $H$ is $K_4$, the complete graph on $4$ vertices. Since $H$ is a Hamiltonian $st$-planar graph, the drawing $\Gamma$ can be constructed as in Theorem~\ref{th:hamiltonian}. Properties (D1)--(D3) are obviously satisfied.

Suppose now that we have a drawing $\Gamma$ of a subgraph~$H$ of~$G$ satisfying Properties~(D1)--(D3). If $H$ is the entire graph~$G$, then $\Gamma$ is the desired drawing of $G$ and we are done, so assume that $H$ is a proper subgraph of $G$. We prove that we can insert in $\Gamma$ a drawing of a vertex of $G$ not in $H$, together with a drawing of its incident edges, obtaining a drawing $\Gamma'$ of a subgraph~$H'$ of~$G$, so that Properties (D1)--(D3) are still satisfied. 

Let $\mathcal C_H=(v_a,v_b,v_c)$ be the cycle delimiting the outer face of $H$. Assume, without loss of generality, that $\max\{a,b\}<c$. This implies that $c\geq 4$. Indeed, if $c=3$, then $\mathcal C_H$ delimits the outer face of $G$, which contradicts Property (D2) for $H$, given that $H$ is not the entire graph~$G$. Hence, $v_c$ was stacked into a face in the construction of $G$ via repeated stacking operations. Since when a vertex is stacked into a face its only neighbors in the graph are the vertices incident to the face and since $\max\{a,b\}<c$, it follows that $v_c$ was inserted into a face delimited by a $3$-cycle comprising $v_a$, $v_b$ and a third vertex $v_d$. Then the subgraph $H'$ of $G$ which are going to draw consists of $H$ and of such a vertex $v_d$.

Property (D1) is satisfied by $H'$ because $H$ is a deep plane $3$-tree and $v_c$ is stacked into the face delimited by $(v_a,v_b,v_d)$, hence when the first vertex of $H$ is stacked into the face delimited by $(v_a,v_b,v_c)$, one of the vertices incident to the face, namely $v_c$, was the last stacked vertex in the construction of $H'$ via repeated stacking operations.  

We prove that $H'$ satisfies Property (D2). Let $\mathcal C_{H'}$ denote the cycle delimiting the outer face of~$H'$. Suppose, for a contradiction, that not all the vertices of~$G$ inside~$\mathcal C_{H'}$ belong to~$H'$. By Property (D2) for $H$, all the vertices of $G$ inside $(v_a,v_b,v_c)$ belong to $H$ and hence to $H'$. It follows that $G$ contains a vertex inside cycle $(v_a,v_c,v_d)$ or a vertex inside cycle $(v_b,v_c,v_d)$; suppose the former, as the latter case is analogous. Note that $G$ contains at least one vertex inside the cycle $(v_a,v_b,v_c)$, namely $v_n$. Let $v_x$ be the vertex stacked into the face delimited by $(v_a,v_c,v_d)$ and let $v_y$ be the vertex of $G$ stacked into the face delimited by $(v_a,v_b,v_c)$. Suppose that $x<y$, as the other case is symmetric. This implies that $\max\{a,b,d\}<c<x<y$. Hence, when $v_y$ is stacked into the face delimited by $(v_a,v_b,v_c)$, neither of $v_a$, $v_b$, and $v_c$ is the vertex that was stacked last in the construction of $G$ via repeated stacking operations. This contradicts the fact that $G$ is a deep $st$-plane $3$-tree. 

In order to prove Property (D3), we construct a planar straight-line dominance drawing~$\Gamma'$ of~$H'$ starting from the planar straight-line dominance drawing $\Gamma$ of $H$. Assume that the sink comes immediately after the source in counter-clockwise direction along the outer face of $H$ (hence $\Gamma$ lies above the line through the source and the sink of $H$), as the other case (in which $\Gamma$ lies below the line through the source and the sink of $H$) is symmetric. Also assume that the clockwise order of the vertices along the outer face of $H$ is $v_a$, $v_b$, $v_c$, as the other case is symmetric. The construction distinguishes six cases.

\begin{figure}[t]
    \centering
    \hfill
    \begin{subfigure}{0.48\linewidth}
        \centering
        {\includegraphics[scale=.9,page=1]{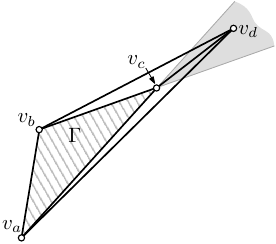}
        \subcaption{\label{fig:deep-1-1}}}
    \end{subfigure}
    \hfill
    \begin{subfigure}{0.48\linewidth}
        \centering
        {\includegraphics[scale=.9,page=2]{Deep.pdf}
        \subcaption{\label{fig:deep-1-2}}}
    \end{subfigure}
    \hfill
    \caption{Construction of a drawing $\Gamma'$ of $H'$ from the drawing $\Gamma$ of $H$ if (a) $v_c$ is the sink of $H$ or (b) $v_c$ is the source of $H$.}
    \label{fig:deep-1}
\end{figure}

In the {\bf first case} (see Figure~\ref{fig:deep-1-1}), $v_c$ is the sink of $H$. Since $v_a$ and $v_b$ have outgoing edges, namely those towards $v_c$, we have that $v_d$ is the sink of $H'$. We place $v_d$ at any point inside the wedge with an angle smaller than $180^{\circ}$ delimited by: (i) the ray starting at $v_c$, lying on the line through $v_a$ and $v_c$, and directed away from $v_a$, and (ii) the ray starting at $v_c$, lying on the line through $v_b$ and $v_c$, and directed away from $v_b$. The constructed drawing $\Gamma'$ is a planar straight-line dominance drawing of $H'$. Indeed, planarity comes from the planarity of~$\Gamma$ and from the placement of $v_d$ in the above wedge. Also, that~$\Gamma'$ is a dominance drawing comes from the fact that $\Gamma$ is a dominance drawing of~$H$, and from the fact that $v_d$ is the sink of $H'$ and is above and to the right of every vertex of $H$.   

In the {\bf second case} (see Figure~\ref{fig:deep-1-2}), $v_c$ is the source of $H$. This case can be handled symmetrically to the previous one.

In the remaining four cases, $v_c$ is neither the source nor the sink of $H$. By the initial assumptions, we have that $v_b$ is the source and $v_a$ the sink of $H$. 

\begin{figure}[t]
    \centering
    \hfill
    \begin{subfigure}{0.48\linewidth}
        \centering
        {\includegraphics[scale=.9,page=3]{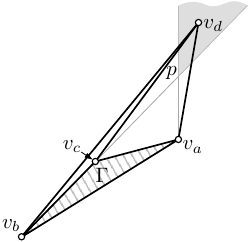}
        \subcaption{\label{fig:deep-2-3}}}
    \end{subfigure}
    \hfill
    \begin{subfigure}{0.48\linewidth}
        \centering
        {\includegraphics[scale=.9,page=4]{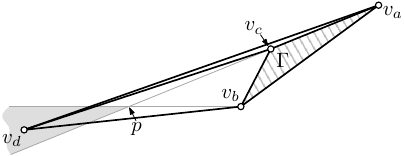}
        \subcaption{\label{fig:deep-2-4}}}
    \end{subfigure}
    \hfill
    \caption{Construction of a drawing $\Gamma'$ of $H'$ from the drawing $\Gamma$ of $H$ if $v_c$ is neither the source nor the sink of $H$ and (a) $v_d$ is the sink of $H'$ or (b) $v_d$ is the source of $H'$.}
    \label{fig:deep-2}
\end{figure}

In the {\bf third case} (see Figure~\ref{fig:deep-2-3}), $v_d$ is the sink of $H'$. Let $p$ be the intersection point between the line through $v_b$ and $v_c$ and the vertical line through $v_a$. We place $v_d$ at any point inside the wedge with an angle smaller than $180^{\circ}$ delimited by: (i) the ray starting at~$p$, lying on the line through $v_b$ and $v_c$, and directed away from $v_b$, and (ii) the ray starting at $p$, lying on the vertical line through $v_a$, and directed away from $v_a$. The proof that the constructed drawing $\Gamma'$ is a planar straight-line dominance drawing of $H'$ is analogous to the one of the first case.   

In the {\bf fourth case} (see Figure~\ref{fig:deep-2-4}), $v_d$ is the source of $H'$. This case can be handled symmetrically to the previous one.

In the remaining two cases, $v_d$ is neither the source nor the sink of $H'$.

\begin{figure}[t]
    \centering
    \hfill
    \begin{subfigure}{0.48\linewidth}
        \centering
        {\includegraphics[scale=.9,page=5]{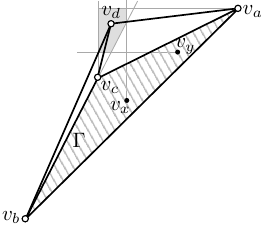}
        \subcaption{\label{fig:deep-3-5}}}
    \end{subfigure}
    \hfill
    \begin{subfigure}{0.48\linewidth}
        \centering
        {\includegraphics[scale=.9,page=6]{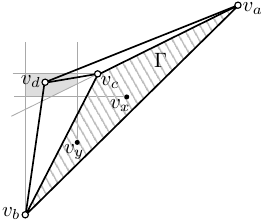}
        \subcaption{\label{fig:deep-3-6}}}
    \end{subfigure}
    \hfill
    \caption{Construction of a drawing $\Gamma'$ of $H'$ from the drawing $\Gamma$ of $H$ if $v_c$ is neither the source nor the sink of $H$, $v_d$ is neither the source nor the sink of $H'$, and (a) the edge between $v_c$ and $v_d$ is directed towards $v_d$ or (b) the edge between $v_c$ and $v_d$ is directed towards $v_c$.}
    \label{fig:deep-3}
\end{figure}

In the {\bf fifth case} (see Figure~\ref{fig:deep-3-5}), the edge between $v_c$ and $v_d$ is directed towards $v_d$. Let~$v_x$ be the leftmost vertex in $\Gamma$ that is different from~$v_c$ and that is not a predecessor of~$v_c$ (possibly~$v_x=v_a$). Also, let~$v_y$ be the highest vertex in~$\Gamma$ that is different from~$v_a$ (possibly~$v_y=v_c$). We place~$v_d$ at any point inside the region~$R$ that is: (i) to the right of the vertical line through~$v_c$; (ii) above the line through~$v_b$ and~$v_c$; (iii) below the horizontal line through~$v_a$; (iv) to the left of the vertical line through~$v_x$; and (v) above the horizontal line through~$v_y$. It is easy to see that~$R$ is a region with interior points. Indeed, any point that is below and to the right of the intersection point between the vertical line through~$v_c$ and the horizontal line through~$v_a$, and sufficiently close to such an intersection point, belongs to~$R$. The constructed drawing $\Gamma'$ is a planar straight-line dominance drawing of $H'$. Indeed, planarity comes from the planarity of $\Gamma$ and from the placement of $v_d$ above $v_c$ and above the line through~$v_b$ and~$v_c$. Also, that~$\Gamma'$ is a dominance drawing comes from the fact that $\Gamma$ is a dominance drawing of~$H$, and from the fact that $v_d$ is in the correct dominance relationship. Indeed, it is below and to the left of $v_a$, which is its only successor in $H'$, it is above and to the right of $v_c$ and of the predecessors of $v_c$, which are its predecessors in $H'$, and it is above and to the left of every other vertex, which is incomparable to it, given that it is above the horizontal line through $v_y$ and to the left of the vertical line through $v_x$.

Finally, in the {\bf sixth case} (see Figure~\ref{fig:deep-3-6}), the edge between $v_c$ and $v_d$ is directed towards $v_c$. This case can be handled symmetrically to the previous one. This concludes the proof.
\end{proof}

\subsection{Sink-dominant $st$-plane $3$-trees}

We next look at the $st$-plane $3$-trees in which every stacking operation happens in a face incident to the sink $t$ of the graph. This results in an $st$-plane $3$-tree in which the sink is adjacent to every vertex. We call this a \emph{sink-dominant $st$-plane $3$-tree}. It is easy to observe that every $n$-vertex maximal $st$-plane graph in which the sink has degree $n-1$ is a sink-dominant $st$-plane $3$-tree (and vice versa). We have the following. 

\begin{theorem} \label{th:sink-dom-3-trees}
	Sink-dominant $st$-plane $3$-trees admit planar straight-line dominance drawings.
\end{theorem}

\begin{proof}
Let $G$ be an $n$-vertex sink-dominant $st$-plane $3$-tree whose outer face is delimited by the cycle $(s,m,t)$. 

\smallskip\noindent{\bf Assumption:} If $n>3$, then let $r$ be  the internal vertex of $G$ adjacent to $s$, $m$ and $t$. If $r$ is a predecessor of $m$, as in Fig~\ref{fig:dominant-assumption}, we add a new source $s'$ adjacent to $s$, $m$ and $t$ in the outer face of $G$, so that the outer face of the resulting graph $G'$ is delimited by the $3$-cycle $(s',s,t)$. Now $m$ is the internal vertex of $G'$ adjacent to $s'$, $s$ and $t$; furthermore, $m$ is a successor of $s$. Hence, by possibly adding a vertex and three edges to $G$ and changing some labels, we can assume, without loss of generality, that the internal vertex $r$ that is adjacent to the three external vertices $s$, $m$ and $t$ of our input $st$-plane $3$-tree $G$ is a successor of $m$. 
	
\smallskip\noindent{\bf Inductive hypothesis:} The proof is similar in spirit to, however more involved than, the proof of Theorem~\ref{th:upper-3-trees}. Let $\Delta$ be any triangle with vertices $p_s$, $p_m$ and $p_t$, where $x(p_s)<x(p_m)<x(p_t)$ and $y(p_s)<y(p_m)<y(p_t)$. If $p_m$ lies above the line through $p_s$ and $p_t$, we say that $\Delta$ is of \emph{type A} (see Fig~\ref{fig:dominant-input-1}), otherwise it is of \emph{type B} (see Fig~\ref{fig:dominant-input-2}). Let $D$ and $E$ be closed disks contained in the interior of $\Delta$ such that:
	\begin{itemize}
		\item if $\Delta$ is of type A, then $D$ and $E$ are horizontally aligned, that is, they have the same two vertical tangents, while if $\Delta$ is of type B, then $D$ and $E$ are vertically aligned;
		\item if $\Delta$ is of type A, then $D$ is strictly below and to the right of $p_m$, while if $\Delta$ is of type B, then $D$ is strictly above and to the left of $p_m$; and
		\item $E$ is strictly above and to the right of $p_m$.
	\end{itemize}
	
	\begin{figure}[t]
		\centering
        \hfill
        \begin{subfigure}{0.32\linewidth}
            \centering
            {\includegraphics[scale=.9,page=6]{Sink-Dominant.pdf}
            \subcaption{\label{fig:dominant-assumption}}}
        \end{subfigure}
        \hfill
        \begin{subfigure}{0.32\linewidth}
            \centering
            {\includegraphics[scale=.9,page=2]{Sink-Dominant.pdf}
            \subcaption{\label{fig:dominant-input-1}}}
        \end{subfigure}
        \hfill
        \begin{subfigure}{0.32\linewidth}
            \centering
            {\includegraphics[scale=.9,page=1]{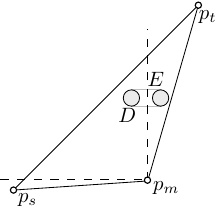}
            \subcaption{\label{fig:dominant-input-2}}}
        \end{subfigure}
        \hfill
		\caption{(a) Augmenting $G$ so that the vertex adjacent to the three vertices on the outer face is a successor of two of them. (b) and (c) Triangle $\Delta$ and disks $D$ and $E$ for the input to the induction. In (b) $\Delta$ is of type A, while in (c) it is of type B.}
		\label{fig:dominant}
	\end{figure}

	We prove, by induction on $n$, that $G$ admits a planar straight-line dominance drawing such that:
	
	\begin{itemize}
		\item $s$ lies at $p_s$, $m$ lies at $p_m$, and $t$ lies at $p_t$; 
		\item every internal vertex of $G$ that is a successor of $m$ lies in the interior of $E$; and 
		\item every internal vertex of $G$ that is incomparable to $m$ lies in the interior of $D$.
	\end{itemize}

	Note that, because of the assumption that $r$ is a successor of $m$, no internal vertex of $G$ is a predecessor of $m$.

	The statement clearly implies the theorem. In the base case, in which $n=3$, the triangle $\Delta$ is the required drawing of $G$ and the statement is true. Suppose now that $n>3$. We only show the construction for the case in which $\Delta$ is of Type B, as the other case is analogous. 

\smallskip\noindent{\bf Graph structure:} Recall that $r$ is the unique vertex of $G$ adjacent to $s$, $m$ and $t$, and that the edge $(m,r)$ is directed towards $r$; refer to Fig~\ref{fig:dominant-structure}.

Let $P_m:=(v_0=m,v_1,\dots,v_\ell=r)$ be the longest directed path from $m$ to $r$. Since every vertex is adjacent to $t$ and since $r$ is a successor of $m$, we have that $P_m$ exists and is unique. For $j=1,\dots,\ell$, let $M_j$ be the subgraph of $G$ induced by the vertices inside or on the boundary of the $3$-cycle $(v_{j-1},v_j,t)$ and note that $M_j$ is a sink-dominant $st$-plane $3$-tree. By the fact that $P_m$ is the longest directed path from $m$ to $r$, we have that, if $M_j$ contains internal vertices, then the internal vertex of $M_j$ that is adjacent to $v_{j-1}$, $v_j$ and $t$ is a successor of $v_j$. This implies that $M_j$ can be drawn recursively.

	Also, let $P_s:=(u_0=s,u_1,\dots,u_k=r)$ be the longest directed path from $s$ to $r$ in $G$ that does not pass through $m$. For $i=1,\dots,k$, the subgraph $S_i$ of $G$ induced by the vertices inside or on the boundary of the $3$-cycle $(u_{i-1},u_i,t)$ is a sink-dominant $st$-plane $3$-tree. Further, if $S_i$ contains internal vertices, then the internal vertex of $S_i$ that is adjacent to $u_{i-1}$, $u_i$ and $t$ is a successor of $u_i$, hence $S_i$ can be drawn recursively. 
	
	Since every vertex of $G$ is adjacent to $t$, the interior of the cycle $\mathcal C_{sm}:=P_s \cup P_m \cup (s,m)$ does not contain any vertices, while it might contain some edges (and in fact it does, unless $\mathcal C_{sm}=(s,m,r)$). We are going to draw $\mathcal C_{sm}$ as a convex curve (hence the edges in its interior will not cause crossings).

\begin{figure}[ht]
		\centering
		\includegraphics[scale=.9,page=3]{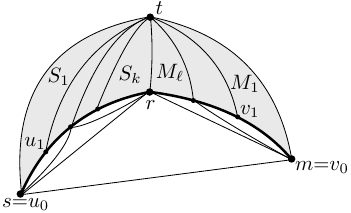}
		\caption{Paths $P_m$ and $P_s$ (as thick lines) and graphs $M_1,\dots,M_\ell,S_1,\dots,S_k$ (with gray interior).}
		\label{fig:dominant-structure}
	\end{figure}

	\begin{figure}[ht]
	\centering
	\includegraphics[scale=1.1,page=4]{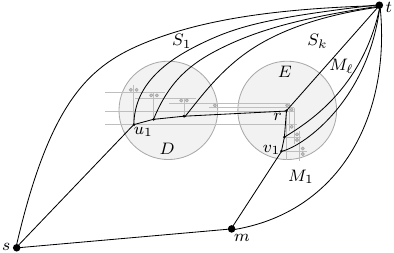}
	\caption{Drawing paths $P_s$ and $P_m$, disks $D^u_1,E^u_1,\dots,D^u_k,E^u_k$, and disks $D^v_1,E^v_1,\dots,D^v_\ell,E^v_\ell$. For the sake of readability, some edges are drawn as curves, as the illustration is mainly meant to represent the relative placement of the vertices of the paths $P_s$ and $P_m$ and of the listed disks.}
	\label{fig:dominant-drawing}
\end{figure}

\smallskip\noindent{\bf Construction:} 	We now draw $P_s$ and $P_m$. We also draw disks inside the triangles representing the cycles $(u_{i-1},u_i,t)$, for $i=1,\dots,k$, and $(v_{j-1},v_j,t)$, for $j=1,\dots,\ell$, so that induction can be applied in order to draw the subgraphs $S_i$ and $M_j$ recursively. Refer to Fig~\ref{fig:dominant-drawing} for an illustration of the relative placement of the vertices of $P_s$ and $P_m$ and of the desired disks.

 	\begin{figure}[ht]
		\centering
		\includegraphics[scale=1.1,page=5]{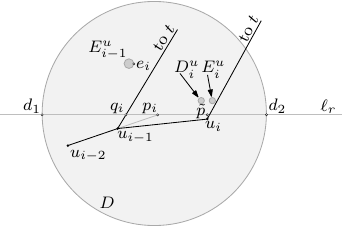}
		\caption{Drawing vertex $u_i$.}
		\label{fig:dominant-detail}
	\end{figure}

We start by placing $r$ at the center of the disk $E$. Next, we draw the vertices $u_1,\dots,u_{k-1}$ in this order inside $D$. Let $\sigma_r$ be the intersection of the horizontal line $\ell_r$ through $r$ with~$D$, and let $d_1$ and $d_2$ be the leftmost and rightmost endpoints of $\sigma_r$, respectively. For $i=1,\dots,k-1$, by drawing $u_i$, we complete the drawing of the triangle $\Delta^u_i$ representing cycle $(u_{i-1},u_i,t)$. Then we also place suitable disks $D^u_i$ and $E^u_i$ inside $\Delta^u_i$ so that $S_i$ can be drawn recursively. 
	
When we have to draw $u_{i}$, for some $i\in \{1,\dots,k-1\}$, we assume that (see Fig~\ref{fig:dominant-detail}):
	
\begin{itemize}
    \item {(C1)} the polygonal line $(u_0,\dots,u_{i-1},r)$ is convex and lies below $\ell_r$;
    \item {(C2)} if $i>1$, the line through $u_{i-2}$ and $u_{i-1}$ cuts $\sigma_r$ in its interior, at a point $p_i$; 
    \item {(C3)} if $i>1$, the segment between $u_{i-1}$ and $t$ cuts $\sigma_r$ in its interior, at a point $q_i$; and 		
    \item {(C4)} if $i>1$, the disks $D^u_{i-1}$ and $E^u_{i-1}$ lie inside $D$ and above $\ell_r$.
\end{itemize} 

We denote by $e_i$ be the rightmost point of $E^u_{i-1}$. Note that conditions (C1)--(C4) are vacuous if $i=1$ (i.e., before drawing $u_1$). In that case, for the sake of simplicity of the description, we let $p_i$, $q_i$, and $e_i$ coincide with $d_1$. 

We now explain how to draw $u_i$.  Let $\overline{x}=\max\{x(p_i),x(q_i),x(e_i)\}$, let $\tilde{x}=(x(d_2)+\overline{x})/2$, where $\overline{x}<\tilde{x}<x(d_2)$, and let $\tilde{p}$ be the point of $\sigma_r$ with $x(\tilde{p})=\tilde{x}$. We place $u_i$ at $(\tilde{x},y(r)-\epsilon)$, where $\epsilon>0$ is sufficiently small so that conditions (C1)--(C3) are satisfied when we have to draw $u_{i+1}$. Indeed, if $\epsilon=0$, then $u_i$ would be placed at $\tilde{p}$ and conditions (C1)--(C3) would be trivially satisfied when we have to draw $u_{i+1}$, hence they are also satisfied for some sufficiently small $\epsilon>0$, by continuity. 

We now place the disks $D^u_i$ and $E^u_i$ so that they have radius $\delta$ and centers at $(\tilde{x}\pm \epsilon', y(r)+\epsilon')$, where $\epsilon'>\delta>0$ are sufficiently small so that:
\begin{itemize}
	\item $D^u_i$ and $E^u_i$ lie inside the triangle $\Delta^u_i=(u_{i-1},u_i,t)$;
	\item $D^u_i$ and $E^u_i$ are lower than $D^u_{i-1}$ and $E^u_{i-1}$; and
	\item condition (C4) is satisfied when we have to draw $u_{i+1}$.
\end{itemize} 
Indeed, if $\epsilon'=\delta=0$ such disks would degenerate and coincide with $\tilde{p}$, which is inside $D$ and also inside $\Delta^u_i$, as the segment $\overline{u_{i-1}t}$ cuts $\sigma_r$ at a point $q_i$ to the left of $\tilde{p}$ and the segment $\overline{u_{i}t}$ cuts $\sigma_r$ at a point to the right of $\tilde{p}$. Hence, such disks remain inside $D$ and $\Delta^u_i$ if $\epsilon'>\delta>0$ is sufficiently small, by continuity. Since $\epsilon'>\delta$, disks $D^u_{i}$ and $E^u_{i}$ lie above $\ell_r$, hence ensuring condition (C4). Finally, choosing $\delta+\epsilon'$ smaller than the distance between $E^u_{i-1}$ and $\ell_r$ ensures that $D^u_i$ and $E^u_i$ are lower than $D^u_{i-1}$ and $E^u_{i-1}$. 


We now draw the vertices $v_1,\dots,v_{\ell-1}$ in this order. For $j=1,\dots,\ell-1$, when we draw $v_j$, we have drawn the triangle $\Delta^v_j$ representing cycle $(v_{j-1},v_j,t)$. Then we also show how to place suitable disks $D^v_j$ and $E^v_j$ inside $\Delta^v_j$ so that $M_j$ can be drawn recursively. 
This is done very similarly to the way vertices $u_1,\dots,u_{k-1}$ and disks $D^u_1,E^u_1,\dots,D^u_{k-1},E^u_{k-1}$ were drawn (see again Fig~\ref{fig:dominant-drawing}), so we only highlight the differences here.

\begin{itemize}
	\item First, all such vertices and disks lie inside $E$, rather than inside $D$. 
	\item Second, the role previously played by $\sigma_r$ is now played by a segment $\sigma'_r$ along the vertical line $\ell'_r$ through $r$. The endpoints of $\sigma'_r$ are the lowest intersection point $e_1$ of $\ell'_r$ with the boundary of $E$ and the intersection point of $\ell'_r$ with the horizontal line through $u_1$. This is because the vertices $v_1,\dots,v_{\ell-1}$ and the disks $D^v_1,E^v_1,\dots,D^v_{\ell-1},E^v_{\ell-1},D^v_{\ell}$ have to be placed below (and to the right of) $u_1,\dots,u_{k-1}$, so to satisfy the constraints of a dominance drawing. When a vertex $v_j$ is drawn, the segment $\overline{v_jt}$ crosses the interior~of~$\sigma'_r$.
	\item Third, the triangles $\Delta^v_j$ are of Type A, unlike the triangles $\Delta^u_i$ which are of Type B. Hence, the disks $D^v_j$ and $E^v_j$ are horizontally aligned, to the right of $\sigma'_r$. The disks $D^v_j$ and $E^v_j$ are above and to the left of the disks $D^v_{j-1}$ and $E^v_{j-1}$.
\end{itemize}

	\begin{figure}[t]
	\centering
	\includegraphics[scale=1.1,page=7]{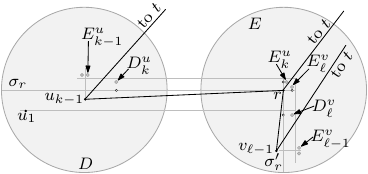}
	\caption{Illustration for the placement of the disks $D^u_{k}$, $E^u_{k}$, $D^v_{\ell}$, and $E^v_{\ell}$. White circles represent initial or intermediate placements for such disks.}
	\label{fig:dominant-final-disks}
\end{figure}

After the vertices $v_1,\dots,v_{\ell-1}$ and the disks $D^v_1,E^v_1,\dots,D^v_{\ell-1},E^v_{\ell-1}$ have been drawn, it only remains to draw the disks $D^u_{k}$ and $E^u_{k}$ inside $\Delta^u_k=(u_{k-1},u_k,t)$ and the disks $D^v_{\ell}$ and $E^v_{\ell}$ inside $\Delta^v_\ell=(v_{\ell-1},v_\ell,t)$. See Fig~\ref{fig:dominant-final-disks}. We have to place $D^u_{k}$ and $E^u_{k}$ above $\sigma_r$ and below $E^u_{k-1}$, with $D^u_{k}$ in $D$ and $E^u_{k}$ in $E$; also, $E^u_{k}$ has to be to the right of $r$. Analogously, we have to place $D^v_{\ell}$ and $E^v_{\ell}$ in $E$, to the right of $\sigma'_r$ and to the left of $E^v_{\ell-1}$, with $E^v_{\ell}$ above $r$. Finally, $E^u_{k}$ has to be above and to the left of $E^v_{\ell}$. 

We can again use continuity arguments to prove that such disk placements exist. Indeed, $\overline{u_{k-1}t}$ cuts the interior of $\sigma_r$, hence $D^u_{k}$ can be initially set to be a point in the interior of $\sigma_r$, to the right of $\overline{u_{k-1}t}$. Analogously, $D^v_{\ell}$ can be initially set to be a point in the interior of $\sigma'_r$ above $\overline{v_{\ell-1}t}$. Disks $E^u_{k}$ and $E^v_{\ell}$ are initially set to coincide with $r$. Now $D^u_{k}$ and $E^u_{k}$ can be moved upward of a sufficiently small distance so that $D^u_{k}$ does not collide with $\overline{u_{k-1}t}$ and remains below $E^u_{k-1}$; note that now $D^u_{k}$ and $E^u_{k}$ are in the interior of $\Delta^u_k$. Analogously, disks $D^v_{\ell}$ and $E^v_{\ell}$ can be moved rightward, of a sufficiently small distance so that they still are to the left of $E^v_{\ell-1}$ and they now both lie in the interior of $\Delta^v_\ell$. Next, we move $E^u_{k}$ rightward and $E^v_{\ell}$ upward so that they are to the right and above $r$, respectively. This movement is sufficiently small so that $E^u_{k}$ remains in $\Delta^u_k$ and $E^v_{\ell}$ in $\Delta^v_\ell$, and so that $E^u_{k}$ remains above and to the left of $E^v_\ell$. Finally, we enlarge the disks so that they have a positive radius. Such a radius can be set to be sufficiently small so that all the above listed properties, which were satisfied before such an enlargement, are still maintained. 

The drawing $\Gamma$ of $G$ is completed by drawing the subgraphs $M_1,\dots,M_\ell$, $S_1,\dots,S_k$ recursively, with triangles $\Delta^v_1,\dots,\Delta^v_\ell,\Delta^u_1,\dots,\Delta^u_k$ representing their outer faces, and with disks  $D^v_1,E^v_1,\dots,D^v_\ell,E^v_\ell,\dots,D^u_1,E^u_1,\dots,D^u_k,E^u_k$ inside such triangles.

\smallskip\noindent{\bf Correctness:} The drawing $\Gamma$ is straight-line by construction. 

The drawings of the subgraphs $M_1,\dots,M_\ell,S_1,\dots,S_k$ are planar by induction. Moreover, the construction guarantees that the cycle $\mathcal C_{sm}$ is represented by a convex curve which keeps in its exterior every edge from a vertex of $\mathcal C_{sm}$ to $t$. It follows that distinct subgraphs among $M_1,\dots,M_\ell,S_1,\dots,S_k$ do not cross each other, that the edges inside or on the boundary of $\mathcal C_{sm}$ do not cross the subgraphs $M_1,\dots,M_\ell,S_1,\dots,S_k$, and that the edges inside or on the boundary of $\mathcal C_{sm}$ do not cross each other. Hence, $\Gamma$ is planar. 

Finally, we prove that $\Gamma$ is a dominance drawing. 
\begin{itemize}
	\item Vertices that are internal to the same subgraph among $M_1,\dots,M_\ell,S_1,\dots,S_k$ are in the correct dominance relationship, by induction. 
    \item Vertices that are internal to distinct subgraphs among $M_1,\dots,M_\ell,S_1,\dots,S_k$ are incomparable. This is because, for any internal vertex $v$ of a subgraph $M_j$ or $S_i$, we have that $t$ is the only vertex incident to the outer face of $M_j$ or $S_i$, respectively, that is a successor of $v$, as a consequence of the fact that $P_s$ and $P_m$ are the longest paths between their end-vertices. By induction, vertices that are internal to distinct subgraphs among $M_1,\dots,M_\ell,S_1,\dots,S_k$ are placed into disks among $D^v_1,E^v_1,\dots,D^v_\ell,E^v_\ell,D^u_1,E^u_1,\dots,D^u_k,$ $E^u_k$. Also, any two disks associated to distinct subgraphs among $M_1,\dots,M_\ell,S_1,\dots,S_k$ are one to the left and above the other one, hence such vertices are in the correct dominance relationship. 
	\item By construction, $v_1,\dots,v_{\ell-1}$ are to the right and below $u_1,\dots,u_{k-1}$, which is the correct dominance relationship as any vertex among $v_1,\dots,v_{\ell-1}$ is incomparable with any vertex among $u_1,\dots,u_{k-1}$.
	\item Also by construction, we have that $v_{j}$ is above and to the right of  $v_{j-1}$, for $j=1,\dots,\ell$, and that $u_{i}$ is above and to the right of  $u_{i-1}$, for $i=1,\dots,\ell$, which is the correct dominance relationship because of the existence of the directed paths  $P_s$ and $P_m$.
	\item Each vertex $u_i$ with $i=1,\dots,k$ is below and to the right of every disk among $D^u_1,E^u_1,\dots,D^u_i$ and is below and to the left of every disk among $E^u_i,D^u_{i+1},\dots,D^u_{k},E^u_{k}$; this is indeed the correct dominance relationship, as all the vertices in the former sequence of disks are incomparable to $u_i$, while all the vertices in the latter sequence of disks are successors of $u_i$. That the vertices among $v_1,\dots,v_{\ell}$ are in the correct dominance relationship with respect to vertices inside disks $D^v_1,E^v_1,\dots,D^v_\ell,E^v_\ell$ can be argued similarly.
	\item Each vertex $u_i$ with $i=1,\dots,k-1$ is above and to the left of every disk among $D^v_1,E^v_1,\dots,D^v_\ell$; this is indeed the correct dominance relationship, as $u_i$ is incomparable to every vertex internal to a subgraph $M_j$ with $j=1,\dots,\ell$, with the exception of the successors of $r$ in $M_\ell$, which are also successors of $u_i$; these vertices are in $E^v_\ell$, which is indeed above and to the right of $u_i$. Similarly, each vertex $v_j$ with $j=1,\dots,\ell-1$ is in the correct dominance relationship with respect to every vertex internal to a subgraph $S_i$ with $i=1,\dots,k$. \qedhere
\end{itemize}    
\end{proof}

Clearly, an analogous result holds true for $st$-plane $3$-trees in which the source is adjacent to every vertex.


\section{Left-non-transitive $st$-plane Graphs} \label{se:non-transitive}

We now consider \emph{left-non-transitive $st$-plane graphs}. These are the $st$-plane graphs such that the left boundary of every face is not a single edge. We show the following.

\begin{theorem} \label{th:left-non-transitive}
Left-non-transitive $st$-plane graphs admit planar straight-line dominance drawings.
\end{theorem}

\begin{proof}
Consider a left-non-transitive $st$-plane graph $G$. We are going to use a \emph{right-to-left path decomposition of $G$}. This consists of a sequence of directed paths $P_1,P_2,\dots,P_k$ such that the following properties are satisfied.
\begin{itemize}
\item $P_1$ is the right boundary of the outer face of $G$;
\item for $i=1,\dots,k$, the graph $G_i:=P_1 \cup P_2 \cup \dots \cup P_i$ is an $st$-plane graph;
\item for $i=2,\dots,k$, the path $P_i$ is the left boundary of a face of $G$ whose right boundary belongs to the left boundary of the outer face of $G_{i-1}$; the internal vertices of $P_i$ do not belong to $G_{i-1}$; and
\item $G_k=G$.
\end{itemize}
This decomposition can be found by ordering the faces of $G$ as in a DFS of the dual of $G$; for a formal proof see, e.g.,~\cite{DBLP:books/sp/Mehlhorn84b}. 

We are going to construct a planar straight-line dominance drawing $\Gamma_i$ of $G_i$, for $i=2,\dots,k$. Then $\Gamma_k$ is the desired drawing of $G$.

\begin{figure}[ht]
	\centering
	\includegraphics[scale=.9]{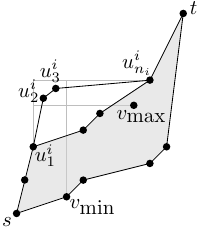}
	\caption{Constructing $\Gamma_i$ from $\Gamma_{i-1}$. The interior of $\Gamma_{i-1}$ is not shown and shaded gray.}
	\label{fig:left-non-transitive}
\end{figure}

For $i=1,\dots,k$, let $P_i=(u^i_1,u^i_2,\dots,u^i_{n_i})$. Since $G$ is left-non-transitive, $n_i \geq 3$ holds. The drawing $\Gamma_1$ of $G_1=P_1$ is constructed as any straight-line drawing such that $x(u^1_1)<x(u^1_2)<\dots<x(u^1_{n_1})$ and $y(u^1_1)<y(u^1_2)<\dots<y(u^1_{n_1})$. Clearly, $\Gamma_1$ is a planar dominance drawing. Now suppose that a planar straight-line dominance drawing $\Gamma_{i-1}$ of $G_{i-1}$ has been constructed, for some $i\in \{2,\dots,k\}$, so that no two vertices have the same $x$- or $y$-coordinate. We construct a planar straight-line dominance drawing $\Gamma_{i}$ of $G_{i}$ from $\Gamma_{i-1}$ as follows; refer to Fig~\ref{fig:left-non-transitive}. Recall that $u^i_1$ and $u^i_{n_i}$ are vertices on the left boundary of $G_{i-1}$, while the internal vertices of $P_i$ need to be inserted into $\Gamma_{i-1}$ in order to define $\Gamma_{i}$. Among all the vertices of $G_{i-1}$ that lie to the right of $u^i_1$ in $\Gamma_{i-1}$, let $v_{\textrm{min}}$ be the one with smallest $x$-coordinate. Also, among all the vertices of $G_{i-1}$ that lie below $u^i_{n_i}$ in $\Gamma_{i-1}$, let $v_{\textrm{max}}$ be the one with largest $y$-coordinate. Note that $x(v_{\textrm{min}}) \leq x(u^i_{n_i})$ and $y(u^i_1) \leq y(v_{\textrm{max}})$. We assign coordinates to the internal vertices of $P_i$ so that $x(u^i_1)<x(u^i_2)<\dots<x(u^i_{n_i-1})<x(v_{\textrm{min}})$ and $y(v_{\textrm{max}})<y(u^i_2)<\dots<y(u^i_{n_i-1})<y(u^i_{n_i})$. This completes the construction of~$\Gamma_{i}$. 

We prove the planarity of $\Gamma_{i}$. Since the drawing of $G_{i-1}$ in $\Gamma_i$ coincides with $\Gamma_{i-1}$ and since $\Gamma_{i-1}$ is planar, it suffices to prove that the edges of $P_i$ do not cross each other and do not cross $\Gamma_{i-1}$. The former follows directly from the fact that $x(u^i_1)<x(u^i_2)<\dots<x(u^i_{n_i-1})<x(u^i_{n_i})$ and $y(u^i_1)<y(u^i_2)<\dots<y(u^i_{n_i-1})<y(u^i_{n_i})$, by construction. We now deal with the latter. 
\begin{itemize}
	\item First, we prove that the edge $(u^i_1,u^i_2)$ does not cross $\Gamma_{i-1}$. Let $(u^i_1,w)$ be the edge of $G_{i-1}$ outgoing from $u^i_1$ and incident to the left boundary of $G_{i-1}$. Such an edge has the outer face of $\Gamma_{i-1}$ to its left, when traversed from $u^i_1$ to $w$. By construction, we have $x(u^i_2)<x(v_{\textrm{min}})\leq x(w)$ and $y(w)\leq y(v_{\textrm{max}})<y(u^i_2)$, hence the interval of $x$-coordinates spanned by the edge $(u^i_1,u^i_2)$ is a subset of the one spanned by the edge $(u^i_1,w)$ and the slope of the edge $(u^i_1,u^i_2)$ is larger than the one of the edge $(u^i_1,w)$. It follows that $(u^i_1,u^i_2)$ lies in the outer face of $\Gamma_{i-1}$, and hence does not cross $\Gamma_{i-1}$.
	\item The proof that the edge $(u^i_{n_i-1},u^i_{n_i})$ does not cross $\Gamma_{i-1}$ is analogous. 
	\item Finally, consider any edge $(u^i_j,u^i_{j+1})$ with $2\leq j\leq n_{i}-2$. By construction, the interval of $x$-coordinates spanned by $(u^i_j,u^i_{j+1})$ is a subset of the interval $(x(u^i_1),x(w))$, where $w$ is defined as above. Also by construction, we have that $y(u^i_1)<y(w)\leq y(v_{\textrm{max}})<y(u^i_j)<y(u^i_{j+1})$. Hence, the edge lies above the edge $(u^i_1,w)$, thus in the outer face of $\Gamma_{i-1}$, which is not crossed by it.
\end{itemize}

We now prove that $\Gamma_{i}$ is a dominance drawing. Since the drawing of $G_{i-1}$ in $\Gamma_i$ coincides with $\Gamma_{i-1}$ and since $\Gamma_{i-1}$ is a dominance drawing, it suffices to prove that the placement of the internal vertices of $P_i$ complies with the dominance relationships they are involved in. Consider any internal vertex $u^i_j$ of $P_i$. For $h=1,\dots,j-1$, vertex $u^i_h$ is a predecessor of $u^i_j$ and indeed we have $x(u^i_h)<x(u^i_j)$ and $y(u^i_h)<y(u^i_j)$, by construction. Analogously, for $h=j+1,\dots,n_i$, vertex $u^i_h$ is a successor of $u^i_j$ and indeed we have $x(u^i_j)<x(u^i_h)$ and $y(u^i_j)<y(u^i_h)$, by construction. Consider any vertex $w$ of $G_{i-1}$ different from $u^i_1$ and $u^i_{n_i}$. 

\begin{itemize}
	\item First, if $w$ is a predecessor of $u^i_1$, then it is also a predecessor of $u^i_j$ and indeed we have $x(w)<x(u^i_j)$ and $y(w)<y(u^i_j)$, given that $x(w)<x(u^i_1)$ and $y(w)<y(u^i_1)$ (since $\Gamma_{i-1}$ is a dominance drawing) and that $x(u^i_1)<x(u^i_j)$ and $y(u^i_1)<y(u^i_j)$ (as proved above). 
	\item Second, if $w$ is a successor of $u^i_{n_i}$, then it is also a successor of $u^i_j$ and indeed we have $x(u^i_j)<x(w)$ and $y(u^i_j)<y(w)$ given that $x(u^i_{n_i})<x(w)$ and $y(u^i_{n_i})<y(w)$ (since $\Gamma_{i-1}$ is a dominance drawing) and that $x(u^i_j)<x(u^i_{n_i})$ and $y(u^i_j)<y(u^i_{n_i})$ (as proved above). 
	\item Finally, if $w$ is neither a predecessor of $u^i_1$ nor a successor of $u^i_{n_i}$, then it is incomparable with $u^i_j$. Note that $x(w)>x(u^i_1)$, as $x(w)<x(u^i_1)$ would imply $y(w)<y(u^i_1)$ (given that $u^i_1$ is on the left boundary of $G_{i-1}$), which is not possible since $w$ is incomparable with $u^i_1$ and $\Gamma_{i-1}$ is a dominance drawing. Analogously, we have $y(w)<y(u^i_{n_i})$. By construction, we have $x(u^i_j)<x(v_{\textrm{min}})\leq x(w)$ and $y(u^i_j)>y(v_{\textrm{max}})\geq y(w)$, hence the placement of $w$ and $u^i_j$ complies with their dominance relationship. 
\end{itemize}

This concludes the proof that  $\Gamma_i$ is a planar straight-line dominance drawing, hence the induction and the proof of the theorem.
\end{proof}

Clearly, an analogous result holds true for \emph{right-non-transitive $st$-plane graphs}, which are $st$-plane graphs such that the right boundary of every face is not a single edge.


\section{$st$-plane Span-$2$ Graphs} \label{se:span-2}

A \emph{level graph} is a directed graph $G = (V,E)$ together with a function $\ell: V\rightarrow \{1,2,\dots,k\}$ such that $\ell(u)<\ell(v)$ for every edge $(u,v)\in E$. We say that an edge $(u,v)$ of $G$ has \emph{span}~$\sigma$ if $\ell(v)-\ell(u)=\sigma$. We call \emph{span-$\sigma$ graph} a level graph such that every edge has span at most~$\sigma$. Span-$1$ graphs are usually called \emph{proper level graphs} and are widely studied in literature. An \emph{$st$-planar level graph} is a level graph $(G,\ell)$ such that $G=(V,E)$ is an $st$-planar graph that admits an upward planar drawing in which $y(u)=\ell(u)$, for every vertex $u\in V$. Note that $st$-planar span-$1$ graphs do not have transitive edges, hence they admit planar straight-line dominance drawings \cite{DBLP:journals/dcg/BattistaTT92}. In this section, we study $st$-planar span-$2$ graphs. Figure~\ref{fig:level-dom-a} depicts one of such graphs, in which each {\em level}, i.e., the set of vertices mapped to the same integer by the function $\ell$, is represented by a dotted line. We prove the following theorem. 


\begin{theorem}
\label{th:span2}
Every $st$-planar span-$2$ graph admits a planar straight-line dominance drawing.
\end{theorem}

\begin{proof}
Let $(G,\ell)$ be an $st$-planar span-$2$ graph. Let $\Psi$ be an upward planar drawing of $G$ in which $y(u)=\ell(u)$, for every vertex $u\in V$. Let~$\mathcal E_G$ be the embedding of~$G$ corresponding to~$\Psi$. We fix $\mathcal E_G$ as the plane embedding of $G$, so that it becomes an $st$-plane span-$2$ graph. The planar straight-line dominance drawing we are going to construct is {\em almost} going to respect~$\mathcal E_G$, as will be formally specified soon. 
 
We define a new $st$-plane span-$2$ graph $(H,\ell')$ as follows. Initialize $(H,\ell')=(G,\ell)$. For each edge $e=(u,z)$ of $G$ with span~$2$, we add two new directed paths $(u,l_{uz},z)$ and $(u,r_{uz},z)$ to $H$. We obtain a plane embedding~$\mathcal E_H$ of~$H$ from~$\mathcal E_G$  so that $e$ is incident to two degree-$3$ faces, one delimited by the path $(u,l_{uz},z)$ and one by the path $(u,r_{uz},z)$. Also, we set $\ell'(l_{uz})=\ell'(r_{uz})=\ell'(z)-1(=\ell'(u)+1)$. See for example the edges $(1,12)$ and $(2,21)$ in Fig.~\ref{fig:level-dom-a}. Note that, in $(H,\ell')$, an edge has span~$1$ if and only if it is non-transitive, and span~$2$ if and only if it is transitive. Also, transitive edges are incident to two faces of degree~$3$. In particular, we are going to use the fact that, for every vertex of $H$, the leftmost and rightmost incoming/outgoing edges are non-transitive.

\begin{figure}[ht]
    \centering
    \begin{subfigure}{0.48\linewidth}
        \centering
        {\includegraphics[page=1, scale=.95]{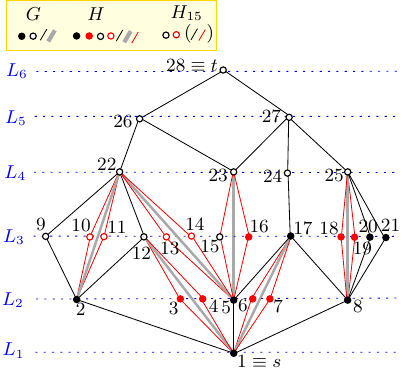}}
        \subcaption[]{}
        \label{fig:level-dom-a}
    \end{subfigure}
    \begin{subfigure}{0.48\linewidth}
        \centering
        {\includegraphics[page=3, scale=.95]{img/level-dominance-revised.pdf}}
        \subcaption[]{}
        \label{fig:level-dom-b}
    \end{subfigure}
    \caption{(a) The augmentation in Theorem~\ref{th:span2}. The initial $st$-plane 2-span graph~$G$ consists of all the black and white vertices, while red-colored vertices and edges are not in~$G$. The graph~$H$ consists of all the depicted vertices and edges. The graph~$H_{15}$ is the subgraph of $H$ induced by the white vertices. (b) The planar straight-line dominance drawing of~$H$ constructed by our algorithm.}
    \label{fig:level-dom}
\end{figure}


Given a planar straight-line dominance drawing $\Gamma_H$ of $H$ (not necessarily respecting~$\mathcal E_H$), it is possible to obtain a planar straight-line dominance drawing $\Gamma_G$ of~$G$ by just removing from $\Gamma_H$ vertices in~$H$ that are not in~$G$, as well as their incident edges. Indeed, for any two vertices $a$ and $b$ of $G$, we have that $a$ is a predecessor of $b$ in~$G$ if and only if $a$ is a predecessor of $b$ in $H$; namely, if a directed path from $a$ to $b$ in $H$ passes through a vertex $l_{uz}$ not in~$G$, then it contains a directed path $(u,l_{uz},z)$, and then the edge $(u,z)$, which belongs to~$G$, can be used in the directed path in place of $(u,l_{uz},z)$. This implies that $\Gamma_G$ is a dominance drawing of~$G$. It remains to prove that $H$ admits a planar straight-line dominance drawing.

Let $\vone$ be the set of vertices of $H$ with one outgoing edge, and let $\vtwo$ be the set of vertices of $H$ with at least two outgoing edges. We define the following partition $\{\vlr,\vl,\vr,\vm\}$ of $\vone$.  Let $v$ be a vertex of $\vone$ and let $e=(v,w)$ be its only incident outgoing edge.

\begin{itemize}
\item If $e$ is the only incoming edge of $w$ in $H$, then $v\in \vlr$.
\item Otherwise, $w$ has more than one incoming edge in $H$. If $e$ is the leftmost or rightmost incoming edge of $w$ in~$\mathcal E_H$, then we have $v\in \vl$ or $v\in \vr$, respectively, otherwise we have $v\in \vm$.
\end{itemize}

See for example \cref{fig:level-dom-a}, where we have $\vtwo=\{1,2,5,8,23\}$, $\vlr=\{17\}$,  $\vl=\{9,15,18,22,26\}$,  $\vr=\{4,14,16,21,25,27\}$, and  $\vm=\{3,6,7,10,11,12,13,19,20,24\}$. 

For $i=1,\dots,k$, let $L_i$ be the sequence of vertices in $H$ defined as follows. First, $L_i$ contains all and only the vertices $v$ such that $\ell'(v)=i$. Second, the vertices in $L_i$ are ordered according to their $x$-coordinates in an upward planar drawing of $H$ that respects~$\mathcal E_H$ and in which $y(u)=\ell'(u)$, for every vertex $u\in V(H)$. See, for example, the sequence $L_2=[2,3,4,5,6,7,8]$ in \cref{fig:level-dom-a}. We define a total order $\vdash$ of the vertices of $H$ so that vertices in $L_i$, with $i\in \{1,\dots,k\}$, are consecutive and ordered as in $L_i$, and so that vertices in $L_{i+1}$ precede vertices in $L_{i}$. For example, in \cref{fig:level-dom-a}, the vertices of $H$ have the following order: $\{28,$ $26,$ $27,$ $22,$  $23,$ $24,$ $25,$  $9,$ $10,$ $11,$  $12,$ $13,$ $14,$  $15,$ $16,$ $17,$  $18,$ $19,$ $20,$ $21,$  $2,$ $3,$ $4,$  $5,$ $6,$ $7,$ $8,$ $1\}$. Let $H_u$ be the subgraph of $H$ induced by $u$ and by every vertex that precedes~$u$ in $\vdash$. For example, in \cref{fig:level-dom-a}, the graph $H_{15}$ is the one induced by the white vertices. Also, let~$\mathcal E_{H_u}$ be the embedding of $H_u$ induced by~$\mathcal E_H$.

For each transitive edge $(u,z)$ of $H$, we define $\mathcal P^{uz} =[p^{uz}_1,p^{uz}_2,\dots,p^{uz}_b]$ and $\mathcal Q^{uz} =[q^{uz}_1,q^{uz}_2,\dots,q^{uz}_c]$ as the maximal sequences of vertices such that:
\begin{itemize}
\item each vertex in $\mathcal P^{uz}\cup \mathcal Q^{uz}$ belongs to $\vm\cap L_{i+1}$ and is a neighbor of both $u$ and $z$;
\item the edges $(p^{uz}_1,z), (p^{uz}_2,z), \dots, (p^{uz}_b,z), (u,z),(q^{uz}_1,z), (q^{uz}_2,z), \dots, (q^{uz}_c,z)$ appear consecutively in this counter-clockwise order around $z$; and 
\item the edge $(p^{uz}_1,z)$ is neither the leftmost edge incoming into $z$ nor the leftmost edge outgoing from $u$, and  the edge $(q^{uz}_c,z)$ is neither the rightmost edge incoming into $z$  nor the rightmost edge outgoing from $u$. 
\end{itemize}

Each of $\mathcal P^{uz}$ and $\mathcal Q^{uz}$ might be empty. For example, $\mathcal P^{uz}$ is empty if the face to the left of $(u,z)$ is a triangular face whose vertex different from $u$ and $z$ is in $\vl$. In \cref{fig:level-dom-a}, we have $\mathcal P^{(2)(22)}=[10]$, $\mathcal Q^{(2)(22)}=[11]$, $\mathcal P^{(8)(25)}=[]$, and $\mathcal Q^{(8)(25)}=[19,20]$. 

For a vertex $w$ of $H$, a planar straight-line dominance drawing $\Gamma_w$ of $H_w$ is \emph{almost-embedding-preserving} ({\sc aep}, for short) if it respects~$\mathcal E_{H_w}$ except that, for each transitive edge $(u,z)$: (i) the edges $(p^{uz}_1,z), (p^{uz}_2,z), \dots, (p^{uz}_b,z), (u,z), (q^{uz}_1,z), (q^{uz}_2,z), \dots, (q^{uz}_c,z)$ appear consecutively and in this order around $z$, except that $(u,z)$ might appear at any place of such a sequence; and (ii) the edges $(u,p^{uz}_1), (u,p^{uz}_2), \dots, (u,p^{uz}_b), (u,z), (u,q^{uz}_1), (u,q^{uz}_2), \dots, (u,q^{uz}_c)$ appear consecutively and in this order around $u$, except that $(u,z)$ might appear at any place of such a sequence.

We now describe how to construct an {\sc aep} planar straight-line dominance drawing $\Gamma$ of~$H$; see \cref{fig:level-dom-b} for an example of a drawing constructed by our algorithm. For any vertex~$u$ of~$H$, the algorithm constructs an {\sc aep} planar straight-line dominance drawing~$\Gamma_u$ of~$H_u$ by augmenting the drawing $\Gamma_{\tilde{u}}$ of $H_{\tilde{u}}$, where $\tilde{u}$ is the vertex immediately preceding $u$ in $\vdash$. Eventually, the algorithm constructs a drawing~$\Gamma$ of~$H$ which coincides with $\Gamma_s$, where $s$ is the source of $H$ and the last vertex in $\vdash$. After describing the construction of $\Gamma_u$, we prove the correctness of the construction, that is, we prove that $\Gamma_u$ is an {\sc aep} planar straight-line dominance drawing.

\smallskip
\noindent
\textbf{Construction:}  We construct $\Gamma_t$, where $t$ is the sink of $H$, by placing $t$ at any point of the plane; note that $t$ is the first vertex in $\vdash$, hence it is the only vertex of $H_t$. Let now $u$ be a vertex of $H$, let $\tilde{u}$ be the vertex immediately preceding $u$ in $\vdash$, and assume that we already constructed $\Gamma_{\tilde{u}}$. We show how to construct $\Gamma_u$. Let $(u,v)$ and $(u,v')$ be the leftmost and rightmost outgoing edges of $u$ in $H$, where $v'=v$ if $u$ has one outgoing edge. Recall that $(u,v)$ and $(u,v')$ are not transitive edges. Also, if $v'=v$, then $u$ is not incident to any transitive edge. We have 5 cases, depending on which of the sets $\vtwo,\vlr,\vl,\vr,\vm$ vertex $u$ belongs to. See Figures \ref{fig:span2-placement14} and \ref{fig:span2-placement5}. 

\begin{figure}[ht]
    \centering
    \begin{subfigure}{0.24\linewidth}
        \centering
        {\includegraphics[page=1, scale=0.6]{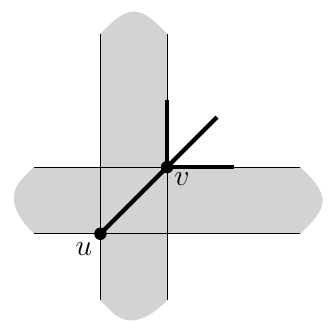}}
        \subcaption[]{}
        \label{fig:span2-placement1}
    \end{subfigure}
    \begin{subfigure}{0.24\linewidth}
        \centering
        {\includegraphics[page=2, scale=0.6]{img/Span2-Placement.pdf}}
        \subcaption[]{}
        \label{fig:span2-placement2}
    \end{subfigure}
    \begin{subfigure}{0.24\linewidth}
        \centering
        {\includegraphics[page=3, scale=0.6]{img/Span2-Placement.pdf}}
        \subcaption[]{}
        \label{fig:span2-placement3}
    \end{subfigure}
    \begin{subfigure}{0.24\linewidth}
        \centering
        {\includegraphics[page=4, scale=0.6]{img/Span2-Placement.pdf}}
        \subcaption[]{}
        \label{fig:span2-placement4}
    \end{subfigure}
    \caption{Placement of $u$ if $u\in \vlr \cup \vl \cup \vr \cup \vm$. Gray areas do not contain any vertex, except, possibly, on the horizontal and vertical lines through $v$. Only edges incident to $v$ are shown or partially shown. (a) $u\in \vlr$. (b) $u\in \vl$. (c) $u\in \vr$. (d) $u\in \vm$.}
    \label{fig:span2-placement14}
\end{figure}

\begin{enumerate}
\item $\mathbf{u\in \vlr}$. We set $x(u)<x(v)$ and $y(u)<y(v)$, where $x(v)-x(u)$ and $y(v)-y(u)$ are sufficiently small so that there is no vertex $w$ such that $x(u)\le x(w)< x(v)$ or with $y(u)\le y(w)< y(v)$. See Figure \ref{fig:span2-placement1} and the placement of $u=17$ with $v=24$ in \cref{fig:level-dom-b}.

\item $\mathbf{u\in \vl}$. We set $y(u)=y(v)$ and $x(u)<x(v)$, where $x(v)-x(u)$ is sufficiently small so that there is no vertex $w$ such that $x(u)\le x(w)< x(v)$. See Figure \ref{fig:span2-placement2} and the placement of $u=9$ with $v=22$, or of $u=22$ with $v=26$ in \cref{fig:level-dom-b}.

\item $\mathbf{u\in \vr}$. We set $x(u)=x(v)$ and we set $y(u)<y(\tilde{u})$, where $y(\tilde{u})-y(u)$ is sufficiently small so that there is no vertex $w$ such that $y(u)\le y(w)< y(\tilde{u})$. See Figure \ref{fig:span2-placement3} and the placement of $u=14$ with $\tilde{u}=13$ and $v=22$, or of $u=21$ with $\tilde{u}=20$ and $v=25$ in \cref{fig:level-dom-b}.

\item $\mathbf{u\in \vm}.$ Note that $u$ and $\tilde{u}$ belong to the same level, that $(\tilde{u},v)$ is an edge of $H_{\tilde u}$, and that $\tilde{u}\in \vl\cup \vm \cup \vtwo$. We set $x(u)$ and $y(u)$ such that $x(w)<x(u)<x(v)$, for any vertex $w$ in $H_{\tilde u}$ such that $x(w)<x(v)$ in $\Gamma_{\tilde u}$, and $y(w)<y(u)<y(\tilde{u})$, for any vertex $w\neq \tilde u$ in $H_{\tilde u}$ such that $y(w)<y(\tilde{u})$ in $\Gamma_{\tilde u}$. See Figure \ref{fig:span2-placement4} and the placement of $u=12$ with $\tilde{u}=11$ and $v=22$, or of $u=20$ with $\tilde{u}=19$ and $v=25$, or of $u=24$ with $\tilde{u}=23$ and $v=27$ in \cref{fig:level-dom-b}.

\item $\mathbf{u\in \vtwo}$.  We set $x(u)=x(v)$ and $y(u)=y(v')$. See Figure \ref{fig:span2-placement5} and the placement of $u=1$ with $v=2$ and $v'=8$, or of $u=5$ with $v=13$ and $v'=17$ in \cref{fig:level-dom-b}.

\begin{figure}[ht]
    \centering
    \includegraphics[page=6, scale=0.75]{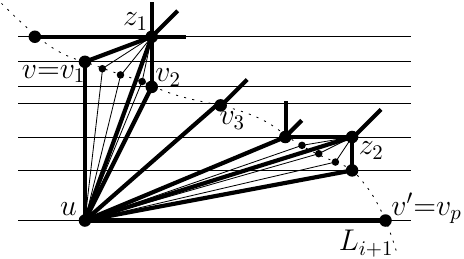}
    \caption{Placement of $u$ if $u\in \vtwo$. Only the edges incident to $u$ and some of the edges incident to the neighbors of $u$ are shown or partially shown. The vertices represented by small disks are those in $\mathcal P^{uz_1}\cup \mathcal Q^{uz_1}$ and $\mathcal P^{uz_2}\cup \mathcal Q^{uz_2}$.}
    \label{fig:span2-placement5}
\end{figure}

If there exist transitive edges outgoing from $u$, we might need to change the position of some already placed vertices in $\vm\cap L_{i+1}$. Namely, for each transitive edge $(u,z)$, if there exists a vertex $w$ in $\mathcal P^{uz}\cup \mathcal Q^{uz}$ along the straight-line segment connecting~$u$ and~$z$, then we move~$w$ to the right. This movement is sufficiently small so that there is no vertex in $\Gamma_u$, other than $w$, whose $x$-coordinate is in the interval $[x_1,x_2]$, where $x_1$ and $x_2$ are the $x$-coordinates of $w$ before and after the movement.
\end{enumerate}

\smallskip \noindent
\textbf{Correctness:} We now prove that $\Gamma_u$ is an {\sc aep} planar straight-line dominance drawing of~$H_u$. This statement, when $u=s$ is the source of $H$, implies the theorem. The proof is by induction on the position of vertex $u$ in $\vdash$. In the base case, $u=t$ is the first vertex in $\vdash$, hence $V(H_t)=\{t\}$, and the statement is trivially true.


For the inductive case, we use the same notation as in the algorithm's description. In particular, $(u,v)$ and $(u,v')$ are the leftmost and rightmost outgoing edges of $u$. Suppose that $\Gamma_{\tilde{u}}$ is an {\sc aep} planar straight-line dominance drawing of~$H_{\tilde{u}}$. We prove that~$\Gamma_u$, constructed as described above, is an {\sc aep} planar straight-line dominance drawing of~$H_u$. The proof distinguishes 5 cases, according to the case of the algorithm which is applied to construct~$\Gamma_u$. In each case, we prove that $\Gamma_u$ is planar and almost-embedding-preserving (\textbf{P}), and that it is a dominance drawing (\textbf{D}). We make the following useful observation. Since~$v\in L_{i+1}$, it follows that every vertex $w$ that lies in the third quadrant of $v$ in $\Gamma_{\tilde{u}}$ (and hence in $\Gamma_{u}$) is a vertex in $L_i$ such that the edge $(w,v)$ exists. This comes from the fact that no vertex in $L_j$ with $j<i$ belongs to $H_{\tilde{u}}$ and that every vertex in the third quadrant of $v$ in~$\Gamma_{\tilde{u}}$ is a predecessor of $v$.


\begin{enumerate}
\item First, suppose that $\mathbf{u\in \vlr}$. 

(\textbf{P}) In order to prove the planarity of $\Gamma_u$, note that $u$ is in the third quadrant of its only successor $v$, by construction. Since $\mathbf{u\in \vlr}$, it follows that $v$ has no incoming edge other than $(u,v)$ in~$H_u$. Since the third quadrant of $v$ in $\Gamma_u$ contains no vertex other than $u$ and since all the edges have non-negative slope in $\Gamma_u$, it follows that the edge $(u,v)$ does not participate in any crossings, hence $\Gamma_u$ is planar, given that $\Gamma_{\tilde{u}}$ is planar. Trivially, $\Gamma_u$ is almost-embedding-preserving, since $\Gamma_{\tilde{u}}$ is almost-embedding-preserving and $(u,v)$ is the only incoming edge of $v$ in $H_u$.

(\textbf{D}) A vertex $w$ is a successor of $u$ if and only if it is a successor of $v$; also, $u$ has no predecessors in $H_u$ and $v$ has no predecessors in $H_u$ other than $u$. By construction, there is no vertex $w$ in $\Gamma_u$ such that $x(u)\le x(w)< x(v)$ or $y(u)\le y(w)< y(v)$, hence the dominance relationship between $u$ and a vertex $w\neq v$ of $H_u$ is the same as the one between $v$ and $w$. Since $\Gamma_{\tilde{u}}$ is a dominance drawing, it follows that $\Gamma_u$ is a dominance drawing, as well.


\item Second, suppose that $\mathbf{u\in \vl}$. 

(\textbf{P}) The proof that $\Gamma_u$ is planar and almost-embedding-preserving is the same as in the case in which $u\in \vlr$. Note that, in this case,~$v$ might have incoming edges other than $(u,v)$ in~$H$, but not in~$H_u$, by the definitions of~$\vl$ and~$\vdash$. 

%
%
%

(\textbf{D}) The argument that proves that $\Gamma_u$ is a dominance drawing is the same as in the case $u\in \vlr$.

%


\item Third, suppose that $\mathbf{u\in \vr}$. 

(\textbf{P}) As in the previous cases, $(u,v)$ does not cross any edge that has no end-vertex in the third quadrant of $v$, given that all the edges have non-negative slope in $\Gamma_u$. However, in this case, $\Gamma_u$ contains vertices other than $u$ in the third quadrant of~$v$, namely all and only the vertices $w$ in $L_i\cap V(H_u)$ such that there exists an edge $(w,v)$. Notice that the edges incoming into $v$ appear in $\Gamma_{\tilde{u}}$ in the same clockwise order around~$v$ as in $\mathcal E_{H_{\tilde{u}}}$, since~$v$ is not the sink of a transitive edge in $H_{\tilde{u}}$, given that it belongs to~$L_{i+1}$ and no vertex of~$L_{i-1}$ is in~$H_{\tilde{u}}$. Among the vertices incident to edges incoming into $v$, there is a vertex~$w^*$ such that~$(w^*,v)$ is the leftmost edge entering~$v$. By definition,~$w^*$ belongs either to~$\vl$ or to~$\vtwo$. In both cases, we have that~$y(w^*)=y(v)$, by construction. Since edges have non-negative slope and~$(u,v)$ has positive slope, no edge outgoing from~$w^*$ might cross~$(u,v)$ in~$\Gamma_u$. Every vertex~$w$ in $L_i$ such that there exists an edge $(w,v)$ and such that $w\notin \{u,w^*\}$ belongs to $\vm$, hence its only outgoing edge is $(w,v)$. By construction $x(w)<x(v)$, hence $(w,v)$ does not cross $(u,v)$ in~$\Gamma_u$, given that $x(u)=x(v)$. It follows that the edge $(u,v)$ does not participate in any crossings, which together with the planarity of $\Gamma_{\tilde{u}}$ implies that $\Gamma_u$ is planar. Also, $\Gamma_u$ is almost-embedding-preserving, given that $\Gamma_{\tilde{u}}$ is almost-embedding-preserving and given that $(u,v)$ is correctly embedded as the rightmost edge incoming into $v$, since $x(u)=x(v)$ and $x(w)<x(v)$, for every vertex $w\neq u$ such that the edge $(w,v)$ belongs to $H_u$.
%
%

(\textbf{D}) Since $\Gamma_{\tilde{u}}$ is a dominance drawing and $u$ has a unique incident edge in $H_u$, in order to prove that $\Gamma_u$ is a dominance drawing it suffices to prove that $u$ is in the correct dominance relationship with every vertex $w$ of $H_u$. By construction, $u$ is in the correct dominance relationship with $\tilde u$, given that such vertices are incomparable and $x(\tilde{u})<x(v)=x(u)$ and $y(u)<y(\tilde{u})$, by construction. Also, consider any vertex $w\notin \{u,\tilde u\}$. If the reachability between $u$ and $w$ is the same as the one between $\tilde u$ and $w$, then $u$ is in the correct dominance relationship with $w$, given that the relative position of $u$ and $w$ is the same as the one of $\tilde{u}$ and $w$, by construction. The only case in which the reachability between $u$ and $w$ is not the same as the one between $\tilde u$ and $w$ is the one in which: (i) $\tilde{u}\in \vtwo$ (hence, $(\tilde{u},v)$ is the leftmost edge incoming into $v$); (ii) $w$ is a successor of $\tilde{u}$; and (iii) $w$ is not a successor of $v$. In fact, in this case $w$ is a successor of $\tilde{u}$, while it is incomparable to $u$. Since $w$ and $v$ are incomparable and since $(\tilde{u},v)$ is the rightmost edge outgoing from $\tilde{u}$, it follows that $x(v)>x(w)$ and $y(v)<y(w)$ in $\Gamma_{\tilde{u}}$ and hence in $\Gamma_u$. Hence, by construction, we have $x(u)=x(v)>x(w)$ and $y(u)<y(v)<y(w)$, hence $u$ is again in the correct dominance relationship with $w$.

\item Fourth, suppose that $\mathbf{u\in \vm}$. 

(\textbf{P}) The proof that~$\Gamma_u$ is planar and almost-embedding-preserving is almost the same as the one for the case in which $u\in \vr$. The only difference is that the reason why $(u,v)$ does not cross a different edge $(w,v)$ is not that $x(u)=x(v)$, but rather that $x(w)\leq x(\tilde{u})<x(u)<x(v)$ and $y(u)<y(\tilde{u})\leq y(w)\leq y(v)$, by construction.



(\textbf{D}) The proof that $\Gamma_u$ is a dominance drawing is the same as in the case in which $u\in \vr$, except that, when $w$ is a successor of $\tilde{u}$ and is incomparable to $u$, we have $x(u)>x(w)$ not because $x(u)=x(v)$, but rather directly by construction.

\item Finally, suppose that $\mathbf{u\in \vtwo}$. 

(\textbf{P}) In order to prove the planarity of $\Gamma_u$, we partition the edges of $H_u$ into four sets. The edges incident to a vertex in $\mathcal P^{uz}\cup \mathcal Q^{uz}$, for some transitive edge $(u,z)$, are called \emph{dangling edges}. The second and third sets of edges are composed of those edges that are not dangling, that are incident to $u$, and that are respectively non-transitive and transitive; these are called \emph{non-transitive $u$-edges} and \emph{transitive $u$-edges}, respectively. Finally, all other edges are called \emph{far edges}. Since $\Gamma_{\tilde u}$ is planar, we only need to prove that the non-far edges do not cross each other and do not cross any far edge. We break down this proof in several parts.

\begin{itemize}
\item First, we prove that non-transitive $u$-edges do not cross each other and that transitive $u$-edges do not cross each other. Let $(u,v=v_1), (u,v_2), \dots, (u,v'=v_p)$ be the clockwise order of the non-transitive $u$-edges in $\mathcal E_{H_u}$. Since $v_1,v_2,\dots,v_p$ belong to $L_{i+1}$, they are incomparable in $H_{\tilde u}$ and $H_u$. Since $\Gamma_{\tilde{u}}$ is almost-embedding-preserving, we have $x(v_1)<x(v_2)<\dots<x(v_p)$ and $y(v_1)>y(v_2)>\dots>y(v_p)$. By construction, we have $x(u)=x(v_1)<x(v_2)<\dots<x(v_p)$ and $y(v_1)>y(v_2)>\dots>y(v_p)=y(u)$, hence no two non-transitive $u$-edges cross each other and the clockwise order of such edges in $\Gamma_u$ is the same as in $H_u$. Similarly, let $(u,z_1),(u,z_2),\dots,(u,z_q)$ be the clockwise order of the transitive $u$-edges in $\mathcal E_{H_u}$. Since $z_1,z_2,\dots,z_q$ belong to~$L_{i+2}$, they are incomparable in $H_{\tilde u}$ and $H_u$. Since $\Gamma_{\tilde{u}}$ is almost-embedding-preserving, we have $x(z_1)<x(z_2)<\dots<x(z_q)$ and $y(z_1)>y(z_2)>\dots>y(z_q)$, hence $x(u)<x(z_1)<x(z_2)<\dots<x(z_q)$ and $y(z_1)>y(z_2)>\dots>y(z_q)=y(u)$, hence any two transitive $u$-edges do not cross each other and their clockwise order in $\Gamma_u$ is the same as in $H_u$.

\item Second, we prove that each non-transitive $u$-edge $(u,v_h)$  does not cross any far edge. If $2\leq h \leq p$, then the third quadrant of~$v_h$ does not contain any vertex other than $u$, since $(u,v_h)$ is the only incoming edge of $v_h$ in $H_u$ (in fact, the only incoming edge of~$v_h$ in $H$ if $2\leq h \leq p-1$, while $v_p=v'$ might have more incoming edges in $H$ that are not in $H_u$) and since any predecessor of $u$ in $H$ belongs to $L_j$ with $j<i$. Then $(u,v_h)$ does not cross any far edge since such an edge has non-negative slope. Finally, consider the edge $(u,v_1)=(u,v)$. Far edges only incident to vertices that are not in the third quadrant of $v$ do not cross $(u,v)$ since they have non-negative slope. Moreover, the only vertices in the third quadrant of $v$ are the vertices $w$ in $L_i$ such that the edge $(w,v)$ exists. For each such a vertex~$w$, the edge~$(w,v)$ does not cross~$(u,v)$ as it lies to the left of it. Also, $w$ can have other incident edges only if $y(w)=y(v)$, and in this case such edges do not intersect $(u,v)$ since they have non-negative slope. 

\item Third, we prove that each transitive $u$-edge $(u,z)$ does not cross any far edge and any non-transitive $u$-edge.  We define two vertices ``associated'' with $(u,z)$. Let $(u,w_z)$ be the edge that follows $(u,p^{uz}_1)$ (or that follows $(u,z)$ if $\mathcal P^{uz}$ is empty) in the counter-clockwise order of the edges outgoing from $u$ in $H_u$; analogously, let $(u,w'_z)$ be the edge that follows $(u,q^{uz}_c)$ (or that follows $(u,z)$ if $\mathcal Q^{uz}$ is empty) in the clockwise order of the edges outgoing from $u$ in $H_u$. For example, in \cref{fig:level-dom-a}, for the transitive edge~$(2,22)$, we have $w_z=9$ and $w'_z=12$, while for the transitive edge $(8,25)$, we have $w_z=18$ and $w'_z=21$. Note that $w_z$ exists, belongs to $L_{i+1}$, and is incident to  the edge $(w_z,z)$; likewise, $w'_z$ exists, belongs to $L_{i+1}$, and is incident to the edge $(w'_z,z)$. We prove the statement for $w_z$, the proof for $w'_z$ is analogous. If $\mathcal P^{uz}$ is empty, then~$w_z$ is the vertex that forms a triangular face with~$(u,z)$ to the left of~$(u,z)$ and the statement directly follows. Otherwise, the statement follows from the assumption that $(p^{uz}_1,z)$ is neither the leftmost edge incoming into $z$ nor the leftmost edge outgoing from $u$, which in fact implies that the edges $(u,w_z)$ and $(w_z,z)$ exist, and thus $w_z$ belongs to~$L_{i+1}$. Also, observe that $(u,w_z)$ is the leftmost edge outgoing from $u$, or $(w_z,z)$ is the leftmost edge incoming into $z$, or both.

We already observed that the clockwise order of the transitive $u$-edges in $\Gamma_u$ and the clockwise order of the transitive $u$-edges in $\Gamma_u$ are the same as in $\mathcal E_{H_u}$. Hence, in order to prove that transitive $u$-edges do not cross non-transitive $u$-edges and that their global clockwise order is the same as in $\mathcal E_{H_u}$, we only need to prove that such orders merge correctly. In order to prove that, it suffices to prove that $(u,w'_z)$ follows $(u,z)$ in clockwise order around $u$ in $\Gamma_u$ and that $(u,w_z)$ precedes $(u,z)$ in clockwise order around $u$ in $\Gamma_u$. We prove the former statement, as the latter has a similar proof. The proof distinguishes two cases. If the edge $(w'_z,z)$ exists then, since $(u,w'_z)$ is not a dangling edge, we have that: (i) $(w'_z,z)$ is the rightmost edge incoming into~$z$, which implies $x(w'_z)=x(z)$ and $y(w'_z)<y(z)$, or (ii) $w'_z=v'$, which implies $x(w'_z)\leq x(z)$ and $y(w'_z)=y(u)<y(z)$. Hence, $(u,w'_z)$ follows $(u,z)$ in clockwise order around~$u$ in~$\Gamma_u$. If the edge $(w'_z,z)$ does not exist, then $w'_z$ and $z$ are incomparable, hence we have $x(z)<x(w'_z)$ and $y(z)>y(w'_z)$ since $\Gamma_{\tilde u}$ is almost-embedding-preserving, and again $(u,w'_z)$ follows $(u,z)$ in clockwise order around $u$ in $\Gamma_u$. 


We next prove that each transitive $u$-edge $(u,z)$ does not cross far edges. Observe that the quadrilateral $R_{uz}=(u,w_z,z,w'_z)$ is strictly-convex in $\Gamma_u$. Indeed, the angles at $u$ and $z$ are at most $90^\circ$ because edge slopes are non-negative, while the angles at $w_z$ and $w'_z$ are smaller than $180^\circ$ because $(u,w_z)$ and $(u,w'_z)$ respectively precede and follow $(u,z)$ in clockwise order around $u$ in $\Gamma_u$. Hence, the edge $(u,z)$ lies inside~$R_{uz}$. Also, the end-vertices of every far edge lie outside (or on the boundary of) $R_{uz}$, given that the only vertices of $H_u$ inside $R_{uz}$ are those vertices in $L_{i+1}$ that have only~$u$ and~$z$ as neighbors. Such vertices are those in $\mathcal P^{uz}$ and $\mathcal Q^{uz}$, hence they are only incident to dangling edges. Moreover, a crossing between $(u,z)$ and an edge $e_{\textrm{far}}$ whose end-vertices are both not inside~$R_{uz}$ would imply the existence of a crossing in $\Gamma_{\tilde{u}}$ or of a crossing between $e_{\textrm{far}}$ and a non-transitive $u$-edge on the boundary of~$R_{uz}$, which we already ruled out. It follows that $(u,z)$ does not cross any far edge. 

\item Finally, recall that dangling edges are those incident to vertices in $\mathcal P^{uz}\cup \mathcal Q^{uz}$, for some transitive $u$-edge $(u,z)$. Since the vertices in $\mathcal P^{uz}\cup \mathcal Q^{uz}\cup \{w_z,w'_z\}$ are pairwise incomparable and since $\Gamma_{\tilde u}$ is an {\sc aep} planar straight-line dominance drawing of $H_{\tilde u}$, it follows that $x(w_z)<x(p^{uz}_1)<x(p^{uz}_2)<\dots<x(p^{uz}_b)<x(q^{uz}_1)<x(q^{uz}_2)<x(q^{uz}_c)<x(w'_z)$ and $y(w_z)>y(p^{uz}_1)>y(p^{uz}_2)>\dots>y(p^{uz}_b)>y(q^{uz}_1)>y(q^{uz}_2)>y(q^{uz}_c)>y(w'_z)$. Since $x(u)\leq x(w_z) <x(w'_z)\leq x(z)$ and $y(z)\geq y(w_z) >y(w'_z)\geq y(u)$, it follows that no two dangling edges incident to vertices in $\mathcal P^{uz}\cup \mathcal Q^{uz}$ cross each other and that they all lie inside $R_{uz}$. This implies that no dangling edge incident to a vertex in $\mathcal P^{uz}\cup \mathcal Q^{uz}$ crosses any non-transitive $u$-edge, or any transitive $u$-edge different from~$(u,z)$, or any far edge, or any dangling edge not incident to a vertex in $\mathcal P^{uz}\cup \mathcal Q^{uz}$, as the end-vertices of all such edges lie outside or on the boundary of $R_{uz}$. Note that a dangling edge incident to a vertex $w$ in $\mathcal P^{uz}\cup \mathcal Q^{uz}$ might cross the transitive $u$-edge~$(u,z)$; however, when this happens, the construction modifies the position of $w$ so that the crossing is avoided.
\end{itemize}

This concludes the proof of planarity of $\Gamma_u$. We now prove that $\Gamma_u$ is almost-embedding-preserving. Since $\Gamma_{\tilde u}$ respects the embedding of $H_{\tilde u}$, we only need to deal with the clockwise order of the edges incident to $u$ and to each vertex adjacent to $u$ in $\Gamma_u$. As argued in the proof of planarity, the clockwise order of the edges incident to $u$ in $\Gamma_u$ is the same as in~$H_u$, with one exception. Namely, for each transitive edge $(u,z)$, the dangling edges incident to $u$ and to vertices in $\mathcal P^{uz}\cup \mathcal Q^{uz}$ are in the same order around~$u$ as in $\mathcal E_{H_u}$, they are correctly inside $R_{uz}$, and the edge $(u,z)$ is also correctly inside $R_{uz}$. However, the edge $(u,z)$ splits the sequence of dangling edges incident to $u$ and incident to vertices in $\mathcal P^{uz}\cup \mathcal Q^{uz}$ in $\Gamma_u$ into two sequences which do not necessarily coincide with the ones in $\mathcal E_{H_u}$. This is allowed by, and is in fact what motivates, the definition of almost-embedding-preserving drawing. Similarly, for each transitive edge $(u,z)$, the dangling edges incident to $z$ are in the same order around $z$ as in $\mathcal E_{H_u}$, however, the edge~$(u,z)$ splits them into two sequences which do not necessarily coincide with the ones in~$\mathcal E_{H_u}$. The clockwise order of the edges incident to each vertex in $\mathcal P^{uz}\cup \mathcal Q^{uz}$, for some transitive edge $(u,z)$, is the same as in~$\mathcal E_{H_u}$, since such vertices have only one incoming and one outgoing edge. Finally, consider each non-transitive $u$-edge $(u,v_h)$, for some $1\leq h\leq p$. If $2\leq h\leq p$, then $(u,v_h)$ is the only edge incoming into $v_h$ in $H_u$; this, together with the fact that $\Gamma_{\tilde u}$ is almost-embedding-preserving, implies that the clockwise order of the edges incident to $v_h$ in $\Gamma_{u}$ is the same as in~$\mathcal E_{H_u}$ (except for the transitive edges outgoing from $v_h$, which might fit into the corresponding sequences of dangling edges differently from~$\mathcal E_{H_u}$). Also, since $x(u)=x(v_1)$, it follows that $(u,v_1)$ is the rightmost edge incoming into $v_1$, as in~$\mathcal E_{H_u}$. Hence, $\Gamma_u$ is almost-embedding-preserving.




(\textbf{D}) Since $\Gamma_{\tilde{u}}$ is a dominance drawing and $u$ has no incoming edge in $H_u$, in order to prove that $\Gamma_u$ is a dominance drawing it suffices to prove that $u$ is in the correct dominance relationship with every vertex $w$ of $H_u$. Since $x(u)=x(v_1)<x(v_2)<\dots<x(v_p)$ and $y(v_1)>y(v_2)>\dots>y(v_p)=y(u)$, we have that $u$ is in the correct dominance relationship in $\Gamma_u$ with all of its successors in $H_u$. Every vertex $w$ of $H_u$ that is not a successor of $u$ is incomparable with $v_p$ and incomparable or a predecessor of $v_1$; this lack of symmetry is due to the fact that $v_p$ does not have predecessors other than $u$ in $H_u$, while $v_1$ might. Since $\Gamma_{\tilde{u}}$ is almost-embedding-preserving, it follows that $w$ is either: (i) below and to the right of $v_p$, or (ii) above and to the left of $v_1$, or (iii) above~$u$, not above $v_1$, and to the left of $v_1$. If $w$ is below and to the right of $v_p$, then we have $x(u)=x(v_1)<x(v_p)<x(w)$ and $y(u)=y(v_p)>y(w)$, hence $u$ is in the correct dominance relationship with $w$. Similarly, if $w$ is above and to the left of~$v_1$, then we have $x(u)=x(v_1)>x(w)$ and $y(u)=y(v_p)<y(v_1)<y(w)$, hence $u$ is in the correct dominance relationship with $w$. Finally, if $w$ is above $u$, not above~$v_1$, and to the left of~$v_1$, the proof that $y(u)=y(v_p)<y(w)$ exploits the fact that $v_p$ and  $w$ are incomparable and that $x(v_p)>x(v_1)>x(w)$.
\end{enumerate}


This concludes the proof that $\Gamma_u$ is an almost-embedding-preserving planar straight-line dominance drawing of $H_u$ and thus the proof of the theorem.
\end{proof}

In the drawing $\Gamma$ constructed in Theorem \ref{th:span2}, a vertex $v$ which is a successor of a vertex $u$ might share the $x$- or $y$-coordinate with $u$. While this is compatible with the definition of dominance drawing, which requires that $x(u)\leq x(v)$ and $y(u)\leq y(v)$, see~\cite{DBLP:journals/dcg/BattistaTT92}, Lemma \ref{le:strict} shows that $\Gamma$ can be modified so that it remains a planar straight-line dominance drawing and no two vertices have the same $x$- or $y$-coordinates.

\section{Conclusions and Open Problems}\label{se:conclusions}

In this paper, we tackled the following problem: Does every $st$-plane graph admit a planar straight-line dominance drawing? While we were not able to solve this question in its generality, our research advanced the state of the art in many directions. 

First, we have provided concrete evidence for the difficulty in constructing planar straight-line dominance drawings. Most notably, we proved that planar straight-line dominance drawings with prescribed $y$-coordinates do not always exist. Our research in this direction indicates that, if an algorithm that constructs a planar straight-line dominance drawing of every $st$-plane graph exists, then it should use substantially different ideas than known algorithms for the construction of upward planar straight-line drawings.

Second, we have described several classes of $st$-plane graphs that admit a planar straight-line dominance drawing. A difficult benchmark here is, in our opinion, provided by the $st$-plane $3$-trees. Hence, we believe it would be a major milestone to understand whether these graphs always admit planar straight-line dominance drawings.

We conclude with one more open problem. Does every (undirected) maximal planar graph admit a planar straight-line dominance drawing? That is, does it admit an $st$-orientation such that the resulting $st$-plane graph has a planar straight-line dominance drawing? 

\paragraph*{\bf Acknowledgments} This research initiated at the First Summer Workshop on Graph Drawing (SWGD 2021). Thanks to the organizers and the participants for an inspiring atmosphere.
\bibliography{Literature}

\end{document}